\documentclass[english,11pt,letter]{article}
\usepackage{geometry}
\usepackage{amsthm} 
\usepackage{wrapfig}

\usepackage{color} % color control for LaTeX documents
\definecolor{note_fontcolor}{rgb}{0.800781, 0.800781, 0.800781}
\definecolor{ForestGreen}{rgb}{0.1333,0.5451,0.1333}

\usepackage{changepage}
\usepackage[nofillcomment,linesnumbered,noend,noresetcount,noline]{algorithm2e}

\AtBeginDocument{%
 \abovedisplayskip=2pt plus 1pt minus 1pt
 \abovedisplayshortskip=0pt plus 1pt
 \belowdisplayskip=2pt plus 1pt minus 1pt
 \belowdisplayshortskip=0pt plus 1pt
}
\usepackage[noindentafter]{titlesec}
\usepackage{amsmath} % enhances the typeset appearance of mathematical formulas
\usepackage{amssymb} % subset of amsfonts that defines the full set of symbol names for two fonts of extra symbols included in amsfonts
\usepackage{graphicx} % enhanced support for graphics
\usepackage{complexity} % computational complexity class names
\usepackage{enumitem} % control layout of itemize, enumerate, description
\usepackage{shorttoc} % table of contents with different depths
\usepackage{array} % extending the array and tabular environments
\usepackage{float}
\usepackage{pdfsync}
\usepackage{verbatim}
\usepackage[pagebackref=true,colorlinks]{hyperref}
\hypersetup{linkcolor=red,filecolor=red,citecolor=ForestGreen,urlcolor=red}

\PassOptionsToPackage{normalem}{ulem}
\usepackage{ulem}

\usepackage{tikz}
\usetikzlibrary{intersections}
\usetikzlibrary{positioning}
\usetikzlibrary{decorations.pathreplacing}
\usetikzlibrary{matrix}
\usetikzlibrary{calc}

\newtheoremstyle{slplain}% name
  {.4\baselineskip\@plus.1\baselineskip\@minus.1\baselineskip}% Space above
  {.3\baselineskip\@plus.1\baselineskip\@minus.1\baselineskip}% Space below
  {\itshape}% Body font
  {}%Indent amount (empty = no indent, \parindent = para indent)
  {\bfseries}%  Thm head font
  {.}%       Punctuation after thm head
  { }%      Space after thm head: " " = normal interword space;
  {}%       Thm head spec
\theoremstyle{slplain} % italics

\newtheorem*{definition*}{Definition}
\newtheorem*{theorem*}{Theorem}
\newtheorem{theorem}{Theorem}[section]
\newtheorem{lemma}[theorem]{Lemma}
\newtheorem{proposition}[theorem]{Proposition}

\newtheorem{corollary}[theorem]{Corollary}
\newtheorem{definition}[theorem]{Definition}

\newenvironment{proofof}[1]{\par{\noindent \it Proof of #1.}}{\qed}

\makeatletter
\newtheorem*{rep@theorem}{\rep@title}
\newcommand{\newreptheorem}[2]{%
\newenvironment{rep#1}[1]{%
 \def\rep@title{#2 \ref{##1}}%
 \begin{rep@theorem}}%
 {\end{rep@theorem}}}
\makeatother

\newreptheorem{theorem}{Theorem}
\newreptheorem{lemma}{Lemma}
\newreptheorem{claim}{Claim}
\newreptheorem{corollary}{Corollary}
\newreptheorem{proposition}{Proposition}

\theoremstyle{definition}

\theoremstyle{plain} % choose from plain/definition/remark

\numberwithin{equation}{section}

\newtheoremstyle{etplain}% name
  {.0\baselineskip\@plus.1\baselineskip\@minus.1\baselineskip}% Space above
  {.0\baselineskip\@plus.1\baselineskip\@minus.1\baselineskip}% Space below
  {\itshape}% Body font
  {}%Indent amount (empty = no indent, \parindent = para indent)
  {\bfseries}%  Thm head font
  {.}%       Punctuation after thm head
  { }%      Space after thm head: " " = normal interword space;
  {}%       Thm head spec

\newcommand{\namedref}[2]{\hyperref[#2]{#1~\ref*{#2}}}

\newcommand{\figurerefb}[2]{\hyperref[#1]{Figure~\ref*{#1}#2}}

\newcommand{\equationref}[1]{\hyperref[#1]{(\ref*{#1})}}
\renewcommand{\eqref}{\equationref}
\newcommand{\adrian}[1]{$\ll$\textsf{\color{red} Adrian: #1}$\gg$}

\newcommand{\DEBUG}[1]{}

\newcommand{\argmax}{\operatornamewithlimits{arg\,max}}
\newcommand{\argmin}{\operatornamewithlimits{arg\,min}}

\renewcommand{\emptyset}{\varnothing}

\renewcommand{\ln}{\log}

\newcommand{\hull}{\text{\texttt{conv}}}

\def\abs#1{\left\vert#1  \right\vert}
\def\norm#1{\left\Vert #1 \right\Vert}

\newcommand\sgn[1]{\textnormal{sgn}(#1)}

\renewcommand\R{\mathbb{R}}

\renewcommand\span{\textnormal{\texttt{span}}}
\newcommand\Cyl{\textnormal{\texttt{Cyl}}}
\newcommand\vol{\textnormal{\texttt{vol}}}

\newcommand\dappx{\delta_{\textnormal{approx}}}

\begin{document}

\linespread{0.97}

% title, authors, etc.
\title{Multidimensional Binary Search for Contextual Decision-Making}

\author{Ilan Lobel\\
NYU\\
\texttt{ilobel@stern.nyu.edu}
\and
Renato Paes Leme\\
Google Research NY\\
\texttt{renatoppl@google.com}
\and 
Adrian Vladu\thanks{Partially supported by NSF grants CCF-1111109 and CCF-1553428}\\
MIT\\
\texttt{avladu@mit.edu}}

\date{}
\maketitle
\begin{abstract}
We consider a multidimensional search problem that is motivated by questions in
contextual decision-making, such as dynamic pricing and personalized medicine.
Nature selects a state from a $d$-dimensional unit ball
and then generates a sequence of $d$-dimensional directions. We are given
access to the directions, but not access to the state. After receiving a direction, we
have to guess the value of the dot product between the state and the direction. Our goal is to minimize the number of times when our guess is more
than $\epsilon$ away from the true answer. We construct a polynomial time
algorithm that we call Projected Volume achieving regret $O(d\log(d/\epsilon))$, which
is optimal up to a $\log d$ factor. The algorithm combines a volume cutting
strategy with a new geometric technique that we call cylindrification.
\end{abstract}

\thispagestyle{empty}

\setcounter{page}{1}
%!TEX root = main.tex

\section{Introduction} 

Binary search is one of the most basic primitives in algorithm design. The binary search
problem consists in trying to guess an unknown real number $\theta \in [0,1]$
given access to an oracle that replies for every guess $x_t$ if $x_t \leq \theta$
or $x_t > \theta$. After $\log(1/\epsilon)$ guesses, the binary search algorithm is able to
estimate $\theta$ within $\epsilon$ precision.

We study a multidimensional and online version of the binary search
problem. The unknown quantity is  a vector $\theta \in \R^d$ with
$\norm{\theta}_2 \leq 1$ and in each iteration an adversary selects a direction
$u_t \in \R^d$ such that $\norm{u_t}_2 = 1$. At each iteration,  the algorithm is asked to guess the
value of the dot product $\theta^\top u_t$. After the algorithm makes a guess $x_t$,
it is revealed to the algorithm whether $x_t \leq \theta^\top u_t$ or $x_t > \theta^\top
u_t$. The goal of the algorithm designer is to create an algorithm that makes as few mistakes as possible, where a
mistake corresponds to a guess with an error larger than $\epsilon$.

This problem has recently come up as a key building block in the design of online algorithms for contextual decision-making. In contextual decision-making, the direction $u_t$ corresponds to a context relevant for the period $t$ decision and $\theta^\top u_t$ corresponds to the optimal period $t$ decision. Contextual decision-making is increasingly important in an economy where decisions are ever more customized and personalized. We now mention two applications:

\vspace{.2cm}
\noindent {\bf Personalized Medicine \cite{bayati2016}}:
 Determining the right dosage of a  drug
for a given patient is a well-studied problem in the medical literature. For example, for
certain anticoagulant drugs, the appropriate dosage can vary by up to a factor of 10
from individual to individual. Increasingly, doctors are using not only demographic information to decide dosage levels, but are also using higher-dimensional clinical and genetic data.  Bastani and Bayati \cite{bayati2016}
propose a mathematical formulation for this problem and tackle it using tools from statistical learning and contextual bandits. At its core, the problem
studied is a multidimensional binary search problem: each patient is associated with a vector of features $u_t$ which describes his or her demographic, genetic and
clinical data. The algorithm outputs a recommended  dosage $x_t$  and then observes only whether the dosage was below or above the ideal level. If the ideal
dosage is a linear function of the features with unknown coefficients $\theta$ then what the algorithm observes is whether $\theta^\top u_t \geq x_t$ or $\theta^\top u_t < x_t$.

\vspace{.2cm}
\noindent {\bf Feature-based Pricing
\cite{amin2014repeated,CohenLL16,qiang2016dynamic,javanmard2016dynamic}}: Consider a firm that 
sells a very large number of differentiated products. Instead of attempting to learn the market value of each product independently, it might be more sensible for the firm to build a pricing model  based on features of each product. In internet
advertising, for example, each impression sold by an ad exchange is unique in its
combination of demographic and cookie data. While it is hopeless for the
exchange to learn how to price each combination in isolation, it is reasonable for the firm
to learn a model of the market value of its different products. In this setting, each product $t$ is described by a vector $u_t$ of features. Assume  the
market value is a linear function with unknown coefficients $\theta$. If the
firm sets a price $x_t$ for this item it will learn that $\theta^\top u_t \geq
x_t$ if the product is sold and that $\theta^\top u_t < x_t$ otherwise. The goal in
this setting is not minimizing guesses that are $\epsilon$ far from the
$\theta^\top u_t$ as in the personalized medicine setting, but to maximize
revenue. Revenue, however, is a very asymmetric objective: if the price is above the
market value we lose the sale and incur a large loss, while underpricing still leads to a sale where the loss in revenue is the difference $\theta^\top u_t - x_t$.
Nevertheless, Cohen et al \cite{CohenLL16} showed that  an algorithm for the
multidimensional binary search problem can be converted into an algorithm for
the feature-based pricing problem in a black-box manner.

\vspace{.2cm}
 The first approach to this problem was due to Amin, Rostamizadeh and Syed
\cite{amin2014repeated} in the context of the pricing problem and
is based on stochastic gradient descent. The stochastic gradient approach
requires the features $u_t$ to be drawn from an unknown iid distribution, so
that each feature can be used to obtain an unbiased estimator of a certain
function. Subsequent approaches by Bastani and Bayati \cite{bayati2016} and
Qiang and Bayati \cite{qiang2016dynamic} use techniques from statistical
learning such as greedy least squares or LASSO. Javanmard and Nazerzadeh
\cite{javanmard2016dynamic} apply a regularized maximum likelihood estimation
approach and obtain an improved regret guarantee. One could also use a general purpose contextual bandit algorithm (such as Agarwal et al. \cite{agarwal2014taming}) to tackle the iid version of the multidimensional binary search problem, but such an algorithm would have regret that is polynomial in $1/\epsilon$ instead of the logarithmic regret obtained by the specialized algorithms.

All the previously discussed work rely on assuming that the directions $u_t$ are sampled iid. The only approach that makes no assumptions about the directions
$u_t$ is by Cohen et al \cite{CohenLL16}. They do so by tackling directly the
multidimensional binary search problem with adversarial feature vectors
$u_t$ and describe an algorithm with a bound of $O(d^2 \log(d/\epsilon))$ on
the number of mistakes. To obtain that bound, the paper builds on the ellipsoid method from convex optimization. The algorithm always maintains a knowledge set in the shape of an ellipsoid and then chooses an $x_t$ that cuts the ellipsoid through its center whenever there is significant uncertainty on the value of  $\theta^\top u_t$. The algorithm then replaces the resulting half-ellipsoid with the smallest ellipsoid that encloses it, and proceeds to the next iteration.

\vspace{.2cm}
\noindent \textbf{Our Contributions:} Our paper significantly improves the regret bound on the multidimensional binary search problem, providing nearly matching upper and lower bounds for this problem. In Proposition \ref{prop:lower-bound}, we construct a lower bound of $\Omega(d\ln(1/\epsilon\sqrt{d}))$ via a reduction to $d$ one-dimensional problems, which is significantly lower than the $O(d^2 \log(d/\epsilon))$ regret bound from Cohen et al \cite{CohenLL16}. 

Under Cohen et al's ellipsoid-based algorithm, a fraction $1-e^{-1/2d}$ of the volume is removed at each iteration. This fraction is significantly less than half because the step of replacing a half-ellipsoid with its own enclosing ellipsoid is expensive in the sense that it adds back to the knowledge set most of the volume removed in the latest cut. Thus, any ellipsoid-based method requires $d$ steps in order to remove a constant fraction of the volume. Therefore, an algorithm that removes a constant fraction of the volume at each iteration has the potential to perform significantly better than an ellipsoid-based method and, thus, might close the gap between the upper and lower bounds. We can thus conjecture that an algorithm that selects $x_t$ in each iteration so as to create two potential knowledge sets of approximately equal volume would perform nearly optimally.

Cutting a convex set into two sets of approximately equal volume is not a difficult task. In a classical result, Gr\"unbaum showed that cutting a set through its centroid generates two sets, each with at least a $1/e$ fraction of the original volume (see Theorem \ref{thm:volume_grunbaum}). Computing a centroid is a $\#$P-hard problem, but finding an approximate value of the centroid is sufficient for our purposes, and an approximation can be computed in polynomial time. An idea similar to this one was proposed by Bertsimas and Vempala~\cite{BertsimasV04}, in a paper where they proposed a method for solving linear programs via an approximate Gr\"unbaum theorem.

However, removing constant fractions of the volume at each iteration is not sufficient for our purposes. Even if the knowledge set has tiny volume, we might not be able to guess the value of $\theta^\top u$ for some directions $u$  with $\epsilon$ accuracy. To solve our problem, we need to ensure that the knowledge set becomes small along all possible directions. An algorithm that does not keep track of the width of the knowledge set along different directions might not perform well. Perhaps surprisingly, our conjecture that an algorithm that cuts through the centroid at each iteration would have near-optimal regret is false. We show in  Theorem \ref{thm:counter-example} that such a centroid algorithm generates a worst-case regret of $\Omega(d^2\ln(1/\epsilon\sqrt{d}))$. This occurs precisely because the centroid algorithm does not keep track of the different widths of the knowledge set. In an ellipsoid-based algorithm, keeping tracks of the widths of a knowledge set is a relatively easy task since they correspond to the eigenvalues of the matrix that represents the ellipsoid. Keeping track of widths is a more difficult task in an algorithm that does not rely on ellipsoids. This brings us to our key algorithmic idea: cylindrification. 

Cylindrification is the technique we introduce of maintaining a set of directions along which the width of the knowledge set is small and expanding the set in those directions, thus converting the set into a high-dimensional cylinder. A cylindrified set when projected onto its subspace of small directions becomes a hypercube. When projected onto its subspace of large directions, a cylindrified set looks exactly like the original set's projection onto the same subspace. Cylindrification reduces regret by significantly increasing the usefulness of each cut.

Our main algorithm, the Projected Volume algorithm, maintains two objects at all times. It maintains a knowledge set (as the previous algorithms did), but it also maintains a set of orthogonal directions along which the knowledge set is small. At each iteration, it cylindrifies the knowledge set and then computes an approximate value of the centroid of the cylindrified set. It then chooses $x_t$ in order to cut through this approximate centroid. In Theorem \ref{thm:centroid}, the main result of our paper, we prove that this algorithm has a near-optimal regret of $O(d\log(d/\epsilon))$.

The analysis of our algorithm relies on a series of results we prove about convex bodies. We first prove a directional version of  Gr\"unbaum's theorem (Theorem \ref{thm:directional_grunbaum}), which states that the width of the two sets along any direction $u$ after a cut through the centroid are at least $1/(d+1)$ of the width along $u$ of the original set. We also prove that Gr\"unbaum's theorem is robust to approximations (Lemma \ref{lemma:approx_grunbaum}) and projections (Lemma \ref{lemma:projected_grunbaum}). We also prove that the process of cylindrification does not add too much volume to the set (Lemma \ref{lemma:cylindrification}). We then use these geometric results to prove that the volume of the knowledge set projected onto its large directions serves as a potential function and show that it decreases exponentially fast, proving our main result.

\vspace{.2cm}
\noindent \textbf{Relation to Standard Online Learning Problems:}
Our problem bears resemblance
with the classic problem in online classification of learning halfspaces with a
margin, which can be solved by the grandfather of all online learning algorithms, the
Perceptron. See~\cite{rosenblatt1958perceptron, minsky1969perceptrons} 
for the original papers and~\cite{littlestone1988learning} for the closely
related Winnow algorithm. Elegant and modern presentations of those can be found
in~\cite{bubeck2011introduction,
shalev2012online} and~\cite{hazan2016introduction}. In this problem, there is an unknown
$\theta$ and, in each iteration, we are given a vector $u_t$ and are asked to guess
the sign of the dot product $\sgn{\theta^\top u_t}$. If we are guaranteed that
all feature vectors $u_t$ are far enough from the separating hyperplane (i.e.,
there is a margin) we can bound the total number of mistakes the algorithm makes.

Both multidimensional binary search and learning halfspaces with a margin have a
similar feedback: which side we are from an unknown hyperplane each point is.
This begs the question of whether the techniques
developed for learning halfspace with a margin or similar online classification
problems can be applied to multidimensional
binary search. There is a subtle difference between the two problems: in
multidimensional binary search we don't observe if a mistake occurred or not.
This feature is crucial for the pricing application that motivates the problem:
in the pricing application, we do not get feedback of whether the price was barely
below the buyer's valuation or much below. We argue in Appendix
\ref{sec:comparison} that this subtle difference poses a significant obstacle to applying
the techniques from one problem to the other.

%!TEX root = main.tex

\vspace{-.2cm}
\section{The Model}

We consider an infinite horizon game between a player and nature. The game
begins with nature selecting a state $\theta$ from the $d$-dimensional unit ball
centered at the origin. We label this ball $K_0$, i.e., $K_0 = \{\theta \in
\mathbb R^d:~ ||\theta||_2 \leq 1\}$. The player knows $K_0$, but does not know
the value of $\theta$. \footnote{Although we assume for
simplicity that $K_0$ is a ball throughout our paper, we could have let $K_0$ be an arbitrary convex body contained inside the unit ball.}

At every period $t=0,1,2,...$, nature selects a vector $u_t$ from the
$d$-dimensional unit sphere, i.e., $U = \{u \in \mathbb R^d:~ ||u||_2 = 1\}$,
that we refer to as the period $t$ direction.  At every period, after nature
reveals $u_t$, the player must choose an action $x_t \in \mathbb R$. The
player's goal is to choose a value of $x_t$ that is close to $u_t^\top \theta$.
Formally, we try to minimize the number of mistakes we make, where a mistake
occurs whenever  $|x_t-u_t^\top \theta| > \epsilon$ for a given $\epsilon > 0$.
We incur regret in period $t$ whenever we make a mistake: \[
  r_t = \begin{cases}
0 & \text{ if } \quad |x_t-u_t^\top \theta| \leq \epsilon~;\\
1 & \text{ if } \quad|x_t-u_t^\top \theta| > \epsilon~.\\
\end{cases}\]

At the end of each period, nature reports to the player whether $x_t \leq u_t^\top \theta$ or $x_t > u_t^\top \theta$. We note that we do not learn the regret $r_t$ in each period, only whether $x_t - u_t^\top\theta$ is positive. Our goal is to find a policy that minimizes our total regret, or equivalently, the total number of mistakes we make over an infinite time horizon, i.e., $R = \sum_{t=1}^\infty r_t$.%\adrian{is this supposed to be infinite? we usually specify regret w.r.t the number of periods}

%\subsection{Applications}\label{sec:pricing}

%\ilan{Explain connections to feature-based pricing (\cite{CohenLL16}, \cite{nazerzadeh2016}) and personalized medicine (\cite{bayati2016}).}

%!TEX root = main.tex

\vspace{-.2cm}
\section{Lower Bound}\label{sec:lower_prior}

  We now construct a lower bound on the regret incurred by our algorithm. The
  lower bound is obtained via a straightforward reduction to $d$ one-dimensional  problems.
	
\begin{proposition}\label{prop:lower-bound} Any algorithm will generate regret
of at least $\Omega(d\ln(1/\epsilon\sqrt{d}))$.  \end{proposition}

\begin{proof} Assume nature selects $\theta$ from within a $d$-dimensional cube
with sides of length $1/\sqrt{d}$. This is a valid choice since the unit ball
$K_0$ contains such a cube. Let $e_i$ represent the vector with value 1 in
coordinate $i \in \{1,...,d\}$ and value 0 in all other coordinates. Suppose
nature selects directions that correspond to the vectors $e_i$ in round-robin
fashion, i.e., $u_t = e_{(t \bmod d) + 1}$. Because of the symmetry of the cube from
which $\theta$ is selected, and the orthogonality of the directions $u_t$, this
problem is equivalent to $d$ independent binary searches over one-dimensional
intervals with length $l= 1/\sqrt{d}$. Our result follows since a one-dimensional binary search
over an interval with length $l$ up to precision $\epsilon$ incurs $\Omega(\ln(l/\epsilon))$ mistakes.  \end{proof}

We note that the lower bound above applies even for the iid version of the multidimensional binary search problem, as nature could be given a distribution over $d$ orthogonal direction vectors. Making the problem offline would also not lower the regret, as having advance knowledge of the direction vectors is useless in the instance above.

%!TEX root = main.tex

\vspace{-.2cm}
\section{The Projected Volume Algorithm}\label{sec:centroid-algorithm}

In this section, we describe the central idea for obtaining near-optimal regret.
In the standard single-dimensional binary search algorithm, the error of the
algorithm at any given iteration is proportional to the length of the interval. The length of the interval thus provides a clear measure in which to make progress. In the multi-dimensional case,
there is no global measure of error, but only a measure of error for each direction.
To make this precise, consider a knowledge set $K \subseteq \R^d$ corresponding
to the set of values of $\theta$ that are compatible with what the algorithm
has observed. Given a direction $u$ (i.e., $u$ is a unit vector), the error incurred by the algorithm to predict
the dot product $u^\top \theta$ corresponds to the directional width of $K$ along $u$:
\begin{equation}\label{eq:width} w(K,u) = \max_{x,y \in K} u^\top
(x-y)~.\end{equation}
which is a measure that is particular for direction $u$. Since the algorithm
does not know which directions it faces in future iterations, it must
decrease some measure that implies progress in a more global sense. A natural
such measure is the volume of $K$. However, measuring volume alone might be
misleading. Consider that case where our current knowledge set is the thin
rectangle represented in Figure \ref{fig:cutting_thin_direction}.

\begin{figure}[h]
\centering
  \begin{tikzpicture}
    \draw[thick] (0,0)--(5,0)--(5,.2)--(0,.2)--cycle;
    \draw[red] (-.3,.1)--(5.3,.1);
    \draw[blue] (2.5,-.1)--(2.5,.3);
    \draw (5.4,0)--(5.6,0)--(5.5,0)--(5.5,.2)--(5.6,.2)--(5.4,.2);
    \node[right] at (5.5,.1) {$\epsilon$};
  \end{tikzpicture}
  \caption{Decreasing volume might not lead to progress with respect to width.
  Both horizontal and vertical cuts remove half the volume, but only the vertical cut makes
  progress towards our goal.}\label{fig:cutting_thin_direction}
\end{figure}
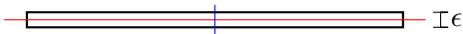

Cutting the knowledge set along either the red horizontal or the blue vertical
line and keeping one of the sides would decrease the volume by half. From the
perspective of our problem, however, the red cut is useless since we already
have a good estimate of the width along that direction. Meanwhile, the blue cut is very useful since it decreases the width along a direction that has still a lot of uncertainty to be resolved.

Motivated by this observation we keep track of the volume of the
knowledge set projected onto a subspace for which there is still a non-trivial
amount of uncertainty. Precisely, our algorithm will be parametrized by a value $\delta
> 0$ which defines the notion of `small'. We maintain two objects:

\begin{enumerate}
  \item the \emph{knowledge set} $K_t \subseteq \R^d$ which will consist of all
    vectors $\theta$ which are consistent with the observations of the algorithm
    so far.
  \item a set of orthonormal vectors $S_t = \{ s_1, \hdots, s_{n_t} \}$ spanning
    a subspace $U_t$ of dimensionality $n_t$ such that the knowledge set has small width
    along any of those directions and has large width along any direction
    perpendicular to them. Formally:
    $$U_t = \span(S_t) \text{ s.t. }w(K_t, s) \leq \delta,~
    \forall s \in S_t
    \text{  and  } w(K_t, u) > \delta, \text{ for all $u$ perpendicular to }
    U_t\,{,}$$
    where $\span(\cdot)$ denotes the span of a set of vectors.  It will be useful to refer to
    $L_t = \{u\vert u^\top s = 0, ~\forall s \in S_t\}$ as the subspace  of large directions.
\end{enumerate}

Our plan will be to ignore a dimension once it becomes small enough and focus on bounding the volume of the projection of the knowledge set $K_t$ onto the subspace of large directions $L_t$. To formalize this notion, let us define the notion of \emph{cylindrification} of a set with respect to orthonormal vectors.

\begin{definition}[Cylindrification] Given a set of orthonormal vectors $S =
  \{s_1, \hdots, s_n\}$, let $L = \{ u\vert u^\top s = 0; \forall s \in S \}$ be the
  subspace orthogonal to $\span(S)$ and $\Pi_L(K)$ be the projection
  \footnote{Formally if $\{\ell_1, \hdots, \ell_k\}$
  is an orthonormal basis of $L$, then $\pi_L(x) = \sum_{i=1}^k \ell_i
  \ell_i^\top x$ and $\Pi_L(K) = \{\pi_L(x)\vert x \in K\}$.  } of $K$ onto $L$.
  Given a convex set $K \subseteq \R^d$ and a
  set of orthonormal vectors $S = \{s_1, \hdots, s_n\}$ we define:
  $$\Cyl(K,S) := \left\{ x + \sum_{i=1}^n y_i s_i\bigg\vert \quad x \in \Pi_L(K) \text{
    and } \min_{\theta \in
  K} \theta^\top s_i \leq y_i \leq \max_{\theta \in K} \theta^\top s_i \right\}~.$$
  Or more concisely, but less intuitively:
  $$\Cyl(K,S) = \Pi_L(K) + \Pi_{\span(s_1)}(K) + \hdots + \Pi_{\span(s_n)}(K) $$
  where the sums applied to sets are Minkowski sums.~\footnote{By Minkowski sum between two sets, we mean $A+B=\{a+b:~ a\in A, b\in B\}$.}
\end{definition}

Informally, the cylindrification operation is designed to create a set with the same projection onto the subspace of large directions, i.e., $\Pi_{L_t}\Cyl(K_t,S_t) = \Pi_{L_t}(K_t)$, while regularizing the projection of the set onto the subspace of small directions: $\Pi_{S_t}\Cyl(K_t,S_t)$ is a box.

We are now ready to present our algorithm, focusing on its geometric aspects and ignoring (for now) the question on how to efficiently compute each step.
The algorithm is parametrized by a constant $\delta > 0$.
It starts with $K_0$ being the ball of radius $1$ and with
$S_0 = \emptyset$. In each iteration the algorithm receives a unit
vector $u_t$ from nature. The algorithm then predicts $x_t$ using the \textit{centroid} $z_t$
of $\Cyl(K_t, S_t)$, by setting $x_t = u_t^\top z_t$. The definition of the centroid is given below:
\begin{definition} The centroid $z$ of a convex set $K$ is defined as
$$z = \frac{1}{\vol(K)} \int_{x\in K} x\, dx\,{,}$$ where $\vol(\cdot)$ denotes the volume of a set.
\end{definition}
Upon learning if the estimate was too small or too large, we update $K_t$ to
$K_{t+1} = K_t \cap \{\theta\vert \theta^\top u_t \leq x_t\}$ or $K_{t+1} = K_t \cap
\{\theta\vert \theta^\top u_t \geq x_t\}$. The next step in our algorithm is
to verify if there exists any direction $v$ orthogonal to $S_t$ such that
$w(K_{t+1}, v) \leq \delta$.
As long as such directions exists, we add them to
$S_t$ and call the resulting set $S_{t+1}$.

Our main result in this paper is:

\begin{theorem}\label{thm:centroid} The Projected Volume algorithm has regret
$O(d\ln(d/\epsilon))$ for the multi-dimensional binary search problem.
\end{theorem}

Our strategy for proving Theorem~\ref{thm:centroid} is to use the volume of the
projection of $K_t$ onto the subspace of large directions as our potential
function:
$$\Phi_t := \vol (\Pi_{L_t}K_t)\,{.}$$
In each iteration, either the set of small directions remains the same or it grows. We first consider the case where the set of small directions remains the same, i.e., $S_{t+1} = S_t$. In this case, we want to argue that the volume of the projection of $K_t$ onto $L_t$ decreases in that iteration. If $S_t = \emptyset$,
then $\Pi_{L_t}K_t = K_t$ and the volume decreases by at least a constant factor. This follows from Gr\"{u}nbaum's Theorem, which we review in the next section. However, if
$S_t \neq \emptyset$, then a decrease in the volume of $K_t$ does not necessarily guarantee a decrease in the volume of the projection.
For example, consider the example in Figure \ref{fig:center_projection}  
where we cut through the center of a rectangular $K_t$. Even though the
volume of $K_{t+1}$ is half the volume of $K_t$, the
volume of the projection onto the $x$-axis doesn't decrease as much. We will
argue that the decrease in volume due to Gr\"{u}nbaum's Theorem extends to
projections (with a small loss) if the width along the cut direction is much
larger than the width along the directions orthogonal to the projection subspace.

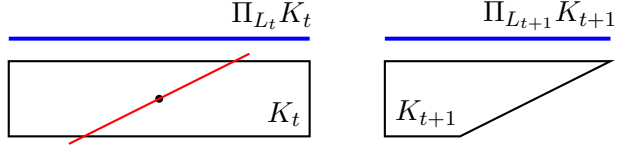
\begin{figure}[h]
  \centering
  \begin{tikzpicture}
    \draw[thick] (0,0)--(4,0)--(4,1)--(0,1)--cycle;
    \fill (2,.5) circle [radius=1.5pt];
    \draw[red, thick] (-1.2+2,-.6+.5)--(1.2+2,.6+.5);
    \draw[blue, line width=1.5pt] (0,1.3)--(4,1.3);
    \node[right] at (3.3,.3) {$K_t$};
    \node at (3.5,1.6) {$\Pi_{L_t} K_t$};

    \begin{scope}[xshift=5cm]
      \draw[thick] (0,0)--(1,0)--(3,1)--(0,1)--cycle;
      \draw[blue, line width=1.5pt] (0,1.3)--(3,1.3);
      \node[right] at (0,.3) {$K_{t+1}$};
      \node at (2.2,1.6) {$\Pi_{L_{t+1}} K_{t+1}$};
    \end{scope}
  \end{tikzpicture}
  \caption{The volume of the projection $\Pi_{L_t}K_t$ decreases slower than the
  volume of $K_t$.}\label{fig:center_projection}
\end{figure}

We now consider the case where we add a new direction to $S_t$. In this case, we will measure the volume in the next iteration as projected onto a subspace of smaller dimension than in period $t$. In general, the volume of a projection can be arbitrarily greater than then volume of the original set. We will use, however, the fact that the $K_t$ is ``large'' along every direction of $L_t$ to argue that adding a vector to $S_t$ can blow up the potential by at most a factor of $O(d^2 / \delta)$. While this is a non-trivial volume increase, this can happen at most $d$ times, leading to a volume increase by a factor of at most $O(d^2/\delta)^d$. We can use this fact to obtain that the algorithm will take at most $O(d \log (d/\delta))$ steps before $L_t$ becomes zero-dimensional.

An inquisitive reader might wonder if we truly need cylindrification to obtain near-optimal regret. We could consider an  algorithm that simply chooses $x_t = u_t^\top z_t$ at each iteration, where $z_t$ is the centroid of $K_t$. We show in Theorem \ref{thm:counter-example} that such an algorithm incurs regret of $\Omega(d^2 \log (1/\epsilon\sqrt{d}))$. Without cylindrification, nature might select directions such that most of the volume reduction corresponds to widths in directions along which the set is already small. Cutting at the centroid of the cylindrified set, instead of the centroid of the original set, is thus crucial to ensure we make progress in the large directions.

The Projected Volume algorithm as discussed above does not actually run in polynomial time since computing the centroid of a convex set is a \#P-hard problem. Fortunately, we can turn Projected Volume into a polynomial time algorithm with a few tweaks, as we show in Theorem \ref{thm:poly-time}. The key step is to approximate the value of the centroid instead of relying on an exact computation. The polynomial time version of Projected Volume presented in Section \ref{sec:computation} also contains a technique for efficiently finding small directions to add to the set $S_t$.% that relies on John's theorem.

%!TEX root = main.tex

\vspace{-.2cm}
\section{Convex Geometry Tools}\label{sec:tools}

In this section, we begin to develop the technical machinery required by the plan outlined
in the previous section. In the heart of the proof will be a statement relating
the volume of a convex body and a volume of its cylindrification with respect to
dimensions along which the body is `small'. In order to obtain this result, we will require
customized versions of Gr\"{u}nbaum's Theorem. Let us start by revisiting the
basic statement of the theorem:

\begin{theorem}[Gr\"{u}nbaum]\label{thm:volume_grunbaum}
  Let $K$ be a convex set, and let $z$ be its centroid. Given an arbitrary nonzero vector $u$, let $K_+ =
  K \cap \{x\vert u^\top (x-z) \geq 0 \}$. Then,
  $$\frac{1}{e} \cdot \vol(K) \leq \vol(K_+) \leq \left( 1- \frac{1}{e} \right)
    \cdot \vol(K)\,{.}$$
\end{theorem}

In other words, any hyperplane through the centroid splits the convex set in two parts, each of which having a constant fraction of the original volume. See
Gr\"unbaum \cite{grunbaum1960} for the original proof of this theorem, or Nemirovski \cite{nemirovski2005efficient}
for a more recent exposition. The first step in the proof of Gr\"{u}nbaum's Theorem consists of applying Brunn's Theorem, which is an immediate consequence of the Brunn-Minkowski inequality:

\begin{theorem}[Brunn]\label{lemma:concave_profile}
Given a convex set $K$, and let $g(t)$ be the $(d-1)$-dimensional volume of the section
$K(t) := K \cap \{x\vert  x^{\top} e_1  = t\}$. 
Then the function $r(t) := g(t)^{1/(d-1)}$ is concave in $t$ over its support.
\end{theorem}

We will rely on Brunn's Theorem to prove our customized versions of Gr\"unbaum's Theorem.\\

\vspace{-.2cm}
\subsection{Directional Gr\"{u}nbaum Theorem}

We begin  by proving a theorem which characterizes how much directional widths of a convex body can change after a cut through the centroid. In some sense, this can be seen as a version of Gr\"unbaum's Theorem bounding widths rather than volumes.

\begin{theorem}[Directional Gr\"{u}nbaum]\label{thm:directional_grunbaum} If $K$ is a convex body
  and $z$ is its centroid,
  then for every unit vector  $u \neq 0$, the set $K_+ =  K \cap \{x\vert u^\top (x-z) \geq 0 \}$ satisfies
  $$\frac{1}{d+1} \cdot w(K, v) \leq w(K_+, v) \leq w(K, v)\,{,} $$
  for all unit vectors $v$.
\end{theorem}

The first step will be to prove Theorem \ref{thm:directional_grunbaum} when $v$
is the direction of $u$ itself. We prove this in the following lemma.

\begin{lemma}\label{lem:dir-grun}
  Under the conditions of Theorem \ref{thm:directional_grunbaum}, $
  w(K_+, u) \geq \frac{1}{d+1} \cdot
  w(K, u)$.
\end{lemma}
We defer the proof of this lemma to Appendix~\ref{appendix:A1}.
We are now ready to prove the Directional Gr\"{u}nbaum Theorem:

\begin{proofof}{Theorem~\ref{thm:directional_grunbaum}}
By translating $K$ we can assume without loss of generality that $z = 0$. Consider three cases:
\begin{enumerate}
\item There exists a point $x_v^+ \in K_+ \cap \argmax_{x \in K} v^\top x$.
In such case, we know by the previous lemma that $$w(K_+, v) \geq v^\top (x_v^+ - z)
\geq \frac{1}{d+1} w(K,v)\,{.}$$
\item The second case is where there exists a point $x_v^- \in K_+
\cap \argmin_{x \in K} v^\top x$. Then,  $$w(K_+, v) \geq v^\top ( z - x_v^-) \geq
\frac{1}{d+1} w(K,v)\,{.}$$
\item In the remaining case, let  $x_v^+ \in \argmax_{x \in K} v^\top x$ and
$x_v^- \in \argmin_{x \in K} v^\top x$ be such that $u^\top x_v^+ < 0$ and
$u^\top x_v^- < 0$. Also, let $x_u = \argmax_{x \in K} u^\top x$.
In such a case, choose real numbers $\lambda^+, \lambda^-$ between zero and one
such that:
$$u^\top \left(x_u + \lambda^+ (x_v^+ - x_u)\right) = 0 \text{  and  }
u^\top \left(x_u + \lambda^- (x_v^- - x_u)\right) = 0\,{.} $$
We can bound $\lambda^+$ and $\lambda^-$ as follows:
$$\frac{1}{d+1} w(K,u) \leq u^\top x_u = \lambda^+ \cdot u^\top (x_u - x_v^+)
\leq \lambda^+ \cdot w(K,u)\,{.}$$
So $\lambda^+ \geq \frac{1}{d+1}$. By the same argument $\lambda^- \geq
\frac{1}{d+1}$. Now, the points, $\tilde{x}_v^+ = x_u +
\lambda^+ (x_v^+ - x_u)$ and $\tilde{x}_v^- = x_u + \lambda^- (x_v^- - x_u)$ are
in $K^+$, since they are convex combinations of points in $K$ and their dot
    product with $u$ is non-negative.    Now:
$$w(K^+, v) \geq v^\top (\tilde{x}_v^+ - \tilde{x}_v^-) = \lambda^+
v^\top(x_v^+ - x_u) + \lambda^- v^\top (x_u - x_v^-) \geq \frac{v^\top (x_v^+ -
x_v^-)}{d+1} = \frac{w(K,v)}{d+1}\,{.}$$
\end{enumerate} 
\end{proofof} 

\vspace{-.2cm}
\subsection{Approximate Gr\"unbaum Theorem}
We will use the Directional Gr\"unbaum Theorem to give an approximate version of
the standard volumetric Gr\"unbaum Theorem. Essentially, we will argue that if
we cut through a point sufficiently close to the centroid, then either side of
the cut will still contains a constant fraction of the volume.

\begin{lemma}[Approximate Gr\"unbaum]\label{lemma:approx_grunbaum}
  Let $K$ be a convex body, and let $z$ be its
  centroid. For an arbitrary unit vector $u$, and scalar $\delta$ such that $0 \leq \delta \leq
  {w(K,u)}/{(d+1)^2}$, let $K_+^\delta = \{x \in K\vert
  u^\top(x-z) \geq \delta\}$. Then, 
  $$\vol(K_+^\delta) \geq \frac{1}{e^2} \cdot \vol(K)\,{.}$$
\end{lemma}
The proof of this lemma follows from a modification of the original proof for Gr\"unbaum's theorem, and it can be found in Appendix~\ref{appendix:A2}.

%!TEX root = main.tex

\vspace{-.2cm}
\section{Cylindrification}\label{sec:cylindrification}

Next we study how to relate the volume of a convex body to the volume of its
projection onto a subspace.

\begin{lemma}[Cylindrification]\label{lemma:cylindrification}
Let $K \subset \R^d$ be a convex body such that $w(K,u) \geq \delta$ for every
unit vector $u$, then for every $(d-1)$ dimensional subspace $L$:
$$\vol(\Pi_L K) \leq \frac{d(d+1)}{\delta} \cdot \vol(K)\,{.}$$
\end{lemma}

As one of the ingredients of the proof, we will use John's Theorem:

\begin{theorem}[John] If $K \subset \R^d$ is a bounded convex body, then there
is a point $z$ and an ellipsoid $E$ centered at the origin such that:
$$z+\frac{1}{d} E \subseteq K \subseteq z + E\,{.}$$
\end{theorem}

In particular, we will use the following consequence of John's Theorem:

\begin{lemma}\label{lemma:large_ball}
If $K \subset \R^d$ is a convex body such that $w(K,u) \geq \delta$
for every unit vector $u$, then $K$ contains a ball of diameter $\delta / d$.
\end{lemma}

\begin{proof}
 Applying John's
theorem and translating $K$ if necessary so that $z=0$, there exists an ellipsoid $E$ such that
$\frac{1}{d} E \subseteq K \subseteq E$. Since the width of $K$ in each
direction is at least $\delta$, the width of $E$ must be at
least $\delta$ in each direction. Since $E$ is an ellipsoid, it must contain a
ball of diameter $\delta$. Thus, $\frac{1}{d}E$ contains a ball of diameter $\frac{\delta}{d}$. Hence, $K$ also contains such a ball.
\end{proof}

\begin{wrapfigure}[12]{r}{0.41\textwidth}
  %\begin{centering}
    \begin{minipage}[c]{0.18\textwidth}

  \begin{tikzpicture}
    \draw[line width=.8pt, fill=lightgray]
    (0,0) \foreach \x in {0,.01,...,2.01} {--(\x,\x*\x-2*\x)} -- (2,0)
    \foreach \x in {0,.01,...,2.01} {--(2-\x,1-.5*\x)} -- cycle;
    \draw[->] (-.2,-1.1)--(2.3,-1.1);
    \fill[black,font=\footnotesize] (2.5,-1.1) node[above] {$\mathbb{R}^{d-1}$};
    \draw[->] (-.1,-1.2)--(-.1,1.2);
    \fill[black,font=\footnotesize] (-.1,1.2) node[right] {$\mathbb{R}$};

    \begin{scope}[shift={(0,-4)}]
      \draw[line width=.8pt, fill=lightgray] (0,0) 
    \foreach \x in {0,.01,...,2.01} {--(\x,-\x*\x+2.5*\x )} -- (2,0) -- cycle;
    \draw[->] (-.2,0)--(2.3,0);
      \fill[black,font=\footnotesize] (2.5,0) node[above] {$\mathbb{R}^{d-1}$};
    \draw[->] (-.1,-.2)--(-.1,2.2);
    \fill[black,font=\footnotesize] (-.1,2.2) node[right] {$\mathbb{R}$};
    \end{scope}

    \draw[dashed, line width=.5pt](.5,-4)--(.5,.25);
    \draw[red, line width=1pt](.5, -4)--(.5, -4 + 1);
    \draw[red, line width=1pt](.5, -.75)--(.5, .25);

  \end{tikzpicture}
    \end{minipage}\hfill
  \begin{minipage}[c]{0.22\textwidth}
  \caption{Illustration of the squashing procedure: the height of all segments
  orthogonal to a $d-1$ subspace is preserved, but the segments are translated
  so to start in the origin of the $d$-th dimension.}\label{fig:squashing}
  \end{minipage}
%\end{centering}
\end{wrapfigure}
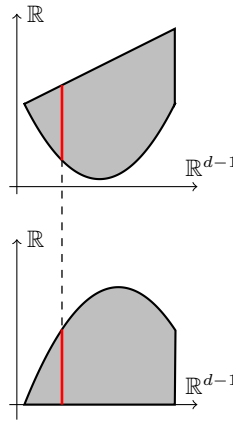

We now prove our cylindrification lemma.

\begin{proofof}{Lemma \ref{lemma:cylindrification}}
Our proof proceeds in two steps:

\emph{Step 1: Squashing $K$.} Assume without loss of generality that the
$(d-1)$-dimensional subspace $L$ is the space defined by the $d-1$ first
coordinates. Represent by $x_L$ the projection of each $x$ onto the $d-1$ first
components
and define $f:\R^{d-1} \rightarrow \R$ such that $f(x_L)$ is the length of
the segment in the intersection of $K$ and the line
$\{(x_L, y):~ y \in \R\}$ (see the top of Figure \ref{fig:squashing}). Formally:
$$f(x_L) = \int_{y \in \R} \mathbf{1}_K (x_L, y) dy\,{.}$$ 
We now argue that $f$ is concave. Given $x_L \in \R^{d-1}$ let $a_x, b_x$ be
such that $(x_L, a_x), (x_L, b_x) \in K$ and $f(x_L) = b_x - a_x$. Let $y_L,
a_y, b_y$ be defined analogously. To see that $f$ is concave, given
$0 < t_x, t_y < 1$ with $t_x + t_y = 1$ we have:
$(t_x x_L + t_y y_L, t_x a_x + t_y a_y)$ and 
$(t_x x_L + t_y y_L, t_x b_x + t_y b_y)$ are in $K$ by convexity, so: $$f(t_x x_L
+ t_y y_L) \leq (t_x b_x + t_y b_y) - (t_x a_x + t_y a_y) = t_x f(x_L) + t_y
  f(y_L)\,{,}$$ which allows us to define the \emph{squashed} version of $K$
  (depicted in the bottom of Figure \ref{fig:squashing}) as:
$$K' = \{(x_L, h):~ x_L \in \Pi_L K,~ 0 \leq h \leq f(x_L) \}\,{.}$$ By construction, $\vol(K') = \vol(K)$. 

\emph{Step 2: Conification.} We know by Lemma \ref{lemma:large_ball}
that $K$ contains a ball of diameter $\delta / d$ so there exists
$x_L$ such that $f(x_L) \geq h := \delta / d$.
Define then the cone $C$ to be the convex hull of $\{(x, 0):~ x \in \Pi_L
K\}$ and $(x_L, h)$. Such cone is a subset of $K'$, so $\vol(K) = \vol(K') \geq \vol(C)$. Since the volume of a $d$-dimensional cone is given by the volume of the base times the height divided by $d+1$,
$$ \vol(K) \geq \vol(C) = \frac{h}{d+1} \cdot \vol(\Pi_L K) = 
\frac{\delta}{d(d+1)} \cdot \vol(\Pi_L K)~.$$
\end{proofof}

%!TEX root = main.tex

\vspace{-.2cm}
\section{Analysis of the Projected Volume Algorithm}

We are now almost ready to analyze our algorithm. To do so, we first consider a version of Gr\"unbaum's Theorem which concerns cuts through the centroid of a cylindrified set. The set being cut is still the original set, but we focus on what happens to the volume of its projection onto the subspace of large directions. The proof of this lemma can be found in Appendix~\ref{appendix:A3}.

\begin{lemma}[Projected Gr\"unbaum]\label{lemma:projected_grunbaum}
  Let $K$ be a convex set contained in the ball of radius $1$, and let
  $S$ be a set of orthonormal
  vectors along which $w(K,s) \leq \delta \leq \frac{\epsilon^2}{16 d(d+1)^2}$,
  for all $s \in S$. Let $L$ be the
  subspace orthogonal to $S$, and let $\Pi_L$ be the projection operator onto that
  subspace.
  If $u$ is a direction along which $w(\Cyl(K,S),u) \geq \epsilon$, $z$ is the
  centroid of the cylindrified body $\Cyl(K,S)$, and
  $K_+ = \{x \in K:~ u^\top (x-z) \geq 0\}$, then:
  $$\vol(\Pi_L K_+) \leq \left( 1-  \frac{1}{e^2} \right) \cdot \vol(\Pi_L K)\,{,}$$
  where $\vol(\cdot)$ corresponds to the $(n-\abs{S})$-dimensional volume on the
  subspace $L$.
\end{lemma}

We now employ the tools we have developed to analyze the regret
of the Projected Volume algorithm. As outlined in Section
\ref{sec:centroid-algorithm}, we will keep in each iteration a convex set $K_t$
of candidate $\theta$ vectors and we will keep a orthonormal basis $S_t$ of
directions for which $K_t$ is small. If $L_t$ is the subspace of directions that
are orthogonal to $S_t$ then our plan is to bound the potential
$\Phi_t = \vol(\Pi_{L_t} K_t)$.
Notice that if $L_t$ is empty, then $S_t$ must be an orthonormal basis such that
$w(K_t, s) \leq \delta, \forall s \in S_t$. In particular, for every unit
vector $u$ and any two $x,y \in K_t$ we must have:
$$u^\top (x-y) = \sum_{s \in S_t} u^\top s \cdot s^\top (x-y) \leq d \delta\,{.} $$
If $\delta \leq \epsilon / d$, then the algorithm will be done once $L_t$
becomes empty. Our goal then is to bound how many iterations can we have where
$L_t$ is non-empty. First we provide a lower bound on the potential. We will use in this section the symbol $\gamma_d$ to denote the volume of the
$d$-dimensional unit ball.
The following loose bound on $\gamma_d$ will be sufficient for our needs: $\Omega(d^{-d})
\leq \gamma_d \leq O(1)$. 

\begin{lemma}\label{lemma:phi-lower} If $L_t$ is non-empty then $\Phi_t \geq
\Omega(\frac{\delta}{d})^{2d}$.
\end{lemma}

\begin{proof}
Let $K_L = \Pi_{L_t} K_t$ and $k$ be the dimension of $L$.
Then $w(K_L, u) \geq \delta$ for all $u \in L$ implies by Lemma
\ref{lemma:large_ball} that $K_L$ contains a ball of radius $\frac{\delta}{k}$,
so $\vol(K_L) \geq \gamma_k \left( \frac{\delta}{k} \right)^k$. Since
$\left( \frac{\delta}{k} \right)^k \geq \left( \frac{\delta}{d} \right)^d$ and
$\gamma_k \geq \Omega(\frac{1}{d})^d$ we have that
$\Phi_t \geq \Omega(\frac{\delta}{d})^{2d}$.
\end{proof}

Now we will give an upper bound on $\Phi_t$ as a function of $t$. Together with
the previous lower bound, we will get a bound on the number of iterations that
can happen before $L_t$ becomes empty. The main ingredient will be a
Gr\"{u}nbaum-type bound on the volume of the projection that is specifically
tailored to our application. For this purpose, we use Lemma~\ref{lemma:projected_grunbaum}, which will specifically address the issue
discussed in Figure \ref{fig:center_projection}. We are now ready for the proof of our main theorem:

\begin{proofof}{Theorem \ref{thm:centroid}}
Our goal is to bound the number of steps for which the algorithm guesses with
at least $\epsilon$ error. Let $R_t$ be the total regret after $t$ steps. Let $N_t$ be $1$ if $w(\Cyl(K_t, L_t), u_t) > \epsilon$ and zero otherwise.
Since $\abs{u_t^\top(z_t - \theta)} \leq \epsilon$ whenever  $w(\Cyl(K_t,
  L_t), u_t) \leq \epsilon$, $R_t \leq \sum_{\tau=1}^t N_\tau$.

Let $K_t$ and $L_t$ be the respective set and subspace after $t$ iterations.
Setting $\delta \leq \frac{\epsilon^2}{16 d(d+1)^2}$ we can apply
Lemma \ref{lemma:projected_grunbaum} directly to obtain that:
  $$\vol(\Pi_{L_t} K_{t+1}) \leq \left( 1 - \frac{1}{e^2} \right)^{N_t} \vol(\Pi_{L_t} K_t)\,{.} $$
If $L_{t+1} = L_t$, then $\vol(\Pi_{L_{t+1}} K_{t+1}) = \vol(\Pi_{L_t} K_{t+1})$.
If we add one new direction $v \in L_t$ to $S$, then we
replace $K_L = \Pi_{L_t} K_t$ by its projection on the subspace $L' = \{x \in L_t:~
v^\top x = 0 \}$. Since $w(K_t, u) \geq \delta, \forall u \in L_t$, then by
Theorem \ref{thm:directional_grunbaum} after we cut $K_t$ we have $w(K_{t+1},u)
\geq \frac{\delta}{d+1}$, so applying the Cylindrification Lemma (Lemma
\ref{lemma:cylindrification}) we obtain:
$$\vol(\Pi_{L'} K_{t+1}) \leq \frac{d(d+1)^2}{\delta} \vol(\Pi_{L_t} K_{t+1})\,{.}$$
If we need to add $r$ new directions to $L_t$ the volume can blow up by at most
$\left( \frac{d(d+1)^2}{\delta} \right)^r$. In particular, since the initial
volume is bounded by $O(1)$, then:
  $$\Omega\left(\frac{\delta}{d}\right)^{2d} \leq
  \Phi_t = \vol(\Pi_{L_t} K_t) \leq O(1) \cdot
\left( \frac{d(d+1)^2}{\delta} \right)^d \cdot \left( 1 - \frac{1}{e^2}
  \right)^{\sum_{\tau=1}^t N_\tau}\,{,}$$
which means that:
\[R_t \leq \sum_{\tau=1}^t N_\tau \leq O\left(d \log \frac{d}{\delta} \right) = O \left(d \log
\frac{d}{\epsilon} \right)~.\]
\end{proofof}

%!TEX root = main.tex
\vspace{-.2cm}
\section{Why Cylindrification?}\label{sec:whycyl}

At the heart of our algorithm lies the simple idea that we should cut a constant fraction of the volume at each iteration if we want to achieve a $\tilde O(d \log(1/\epsilon))$ regret bound. Our algorithm, however, is quite a bit more complex than that. It also keeps a set of `small' directions $S_t$ and it cuts through the center of a cylindrified version of the knowledge set $K_t$ at each iteration. An inquisitive reader might wonder whether this additional complexity is really necessary.  In this section we argue that cylindrification is actually
necessary to obtain our near-optimal regret bound. We prove there exists an instance where the algorithm that only cuts through
the center of the knowledge set (without cylindrifying it first)
incurs $\Omega(d^2 \log(1/\epsilon\sqrt{d}))$ regret.

Formally, consider the algorithm that only keeps $K_t$ and in each iteration
guesses $x_t = u_t^\top z_t$ where $z_t = \frac{1}{\vol(K_t)} \int_{K_t} x dx$
and updates $K_t$ to $K_t^+$ or $K_t^-$. We call this procedure the Centroid algorithm. In order to construct an instance with
$\Omega(d^2 \log(1/\epsilon\sqrt{d}))$ regret for this algorithm, we first define the
following set. Given $s = (s_1, \hdots, s_k)$ with $s_i > 0$ for all $i$, define:
$$\textstyle \Delta(s) = \{x \in \R^k_+:~ \sum_i \frac{x_i}{s_i} \leq 1 \} =
\hull(\{0, s_1 e_1, \hdots, s_k e_k \}),$$
where $\hull(\cdot)$ denotes the convex hull of a set of points.

\begin{lemma}\label{lemma:centroid_simplex}
The centroid of $\Delta(s)$ is given by $\frac{s}{k+1}$.
\end{lemma}

We now consider how the Centroid algorithm performs on a particular set, when
nature selects a specific sequence of directions. The set we start from is the
product between a $(d-k)$-dimensional hypercube and a $k$-dimensional set
$\Delta(s)$, where only the $k^{th}$ entry of $s$ is significantly larger than
$\epsilon$. We now argue that nature might require us to take $\Omega(k \log
(1/\epsilon))$ into a similarly structured set with $k$ replaced by $k+1$.
Repeating this argument $d$ times will lead to our negative conclusion on the
performance of the Centroid algorithm.

\begin{lemma}\label{lemma:step_counterexample}
  Let $1 \leq k < d$, $s \in \R^{k}$ with $0 \leq s_i \leq \epsilon$ for $i <
  k$, $\frac{1}{4} \leq s_k \leq 1$. If $$K = \Delta(y)  \times [0,1]^{d-k}$$
  then there is a sequence of $\Omega(k \log (1/\epsilon))$
  directions $u_t$ such that the Centroid algorithm incurs $\Omega(k \log (\frac{1}{
  \epsilon}))$ regret and by the end of the   sequence, the knowledge set has the form:
  $$K' = \Delta(s') \times [0,1]^{d-k-1}$$
  where $s' \in \R^{k+1}$, $0 \leq s'_i \leq \epsilon$ for $i < k+1$,
  $\frac{1}{4} \leq s'_i \leq  1$.
\end{lemma}

\begin{figure}[h]
\centering
% Sketch output, version 0.3 (build 7d, Wed Oct 19 13:41:05 2016)
% Output language: PGF/TikZ,LaTeX
\begin{tikzpicture}[line join=round]
\filldraw[fill opacity=0.8,draw=red,fill=white](4.761,-1.539)--(8,0)--(5.403,-1.539)--cycle;
\filldraw[fill opacity=0.8,draw=red,fill=white](4.761,-1.539)--(8,0)--(4.499,-.599)--cycle;
\filldraw[draw=red,fill opacity=0.8,fill=white](5.403,-1.539)--(8,0)--(4.499,-.599)--cycle;
\filldraw[fill opacity=0.8,fill=white](4,0)--(3.738,.94)--(-.581,-1.112)--(-.319,-2.052)--cycle;
\filldraw[fill opacity=0.8,fill=white](4,0)--(-.319,-2.052)--(.324,-2.052)--(4.643,0)--cycle;
\filldraw(1.966,-.716) circle (1pt);
\filldraw[fill opacity=0.8,fill=white](4,0)--(4.643,0)--(3.738,.94)--cycle;
\filldraw[fill opacity=0.8,fill=white](0,0)--(-.262,.94)--(-4.581,-1.112)--(-4.319,-2.052)--cycle;
\filldraw[fill opacity=0.8,fill=white](0,0)--(-4.319,-2.052)--(-3.676,-2.052)--(.643,0)--cycle;
\filldraw[fill opacity=0.8,fill=white](0,0)--(.643,0)--(-.262,.94)--cycle;
\filldraw[fill opacity=0.8,fill=white](4.643,0)--(.324,-2.052)--(-.581,-1.112)--(3.738,.94)--cycle;
\filldraw[fill opacity=0.8,fill=white](.643,0)--(-3.676,-2.052)--(-4.581,-1.112)--(-.262,.94)--cycle;
\filldraw[fill opacity=0.8,draw=red,fill=white](4.761,-1.539)--(3.681,-2.052)--(4.324,-2.052)--(5.403,-1.539)--cycle;
\filldraw[fill opacity=0.8,draw=red,fill=white](4.761,-1.539)--(4.499,-.599)--(3.419,-1.112)--(3.681,-2.052)--cycle;
\filldraw[fill opacity=0.8,draw=red,fill=white](5.403,-1.539)--(4.324,-2.052)--(3.419,-1.112)--(4.499,-.599)--cycle;
\filldraw[fill opacity=0.8,draw=red,fill=white](3.681,-2.052)--(3.419,-1.112)--(4.324,-2.052)--cycle;
\filldraw[fill opacity=0.8,fill=white](-.319,-2.052)--(-.581,-1.112)--(.324,-2.052)--cycle;
\filldraw[fill opacity=0.8,fill=white](-4.319,-2.052)--(-4.581,-1.112)--(-3.676,-2.052)--cycle;
\fill[black,font=\footnotesize]
                (-4.663,-1.742) node {$s_1$}
								(-4.041,-2.287) node {$s_2$}
								(-2.474,.102) node {$1$};
\filldraw[draw=red,fill opacity=0,fill=white](1.403,-1.539)--(4,0)--(.499,-.599)--cycle;
\filldraw(1.966,-.716) circle (1pt);
\draw[arrows=->](2.357,-.401)--(2.57,-.23);
\draw[arrows=-](1.966,-.716)--(1.972,-.711);
\draw[arrows=-](1.972,-.711)--(2.357,-.401);
\end{tikzpicture}
% End sketch output
\caption{Illustration of Step 2 in the proof of Lemma
\ref{lemma:step_counterexample} for $k=2$ and $d=3$.}\label{fig:lower_1}
\end{figure}
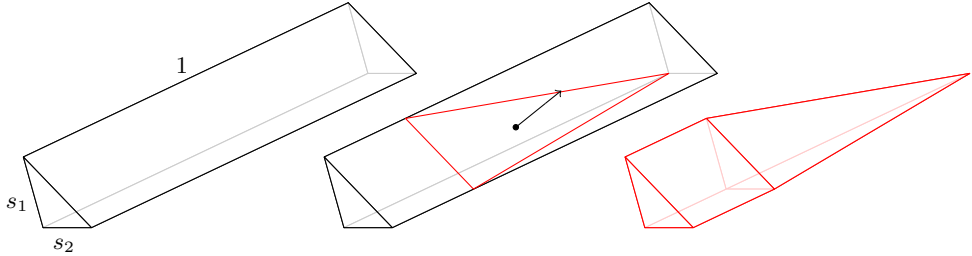

\begin{proof}[Proof sketch]
  Starting from $K$, as a first step we select $\Omega(k
  \log(\frac{1}{\epsilon}))$ vectors in the direction $e_k$ which cause the side
  of the $k$-th side of the simplex $\Delta(s)$ to reduce to $\epsilon$ getting
  one unit of regret in each step. At the end of this step we are left with the
  situation illustrated in Figure \ref{lemma:step_counterexample}. In step $2$
  we choose a direction slightly bent towards  $e_{k+1}$ to carve a $(k+1)$-dimensional simplex out of $K$. The resulting
  shape will be, as depicted in Figure \ref{lemma:step_counterexample}, only
  partially what we want. In the third step we select more directions along
  $e_{k+1}$ to remove the `leftover' and keep only the part corresponding to the
  $(k+1)$-dimensional simplex. A complete proof is provided in Appendix
  \ref{appendix:proof_step_counterexample}.
\end{proof}

This is the main ingredient necessary to show that the algorithm without cylindrification
can incur $\Omega(d^2 \log(1/\epsilon\sqrt{d}))$ regret.

\begin{theorem}\label{thm:counter-example}
The algorithm that always chooses $x_t = u_t^\top z_t$ where $z_t$ is the
centroid of $K_t$ can incur $\Omega(d^2 \log (1/\epsilon\sqrt{d}))$ regret.
\end{theorem}

\begin{proof}
Start with the set $K_0 = [0,1]^d$. Apply Lemma \ref{lemma:step_counterexample} for $k=1, 2, 3,
\hdots, d-1$. The total regret is $\sum_{k=1}^{d-1} \Omega(k
\log(\frac{1}{\epsilon})) = \Omega(d^2 \log(\frac{1}{\epsilon}))$. To construct a valid instance  (one that fits within a ball of radius 1), we replace our initial set with $K_0 = \left[0,1/\sqrt{d}\right]^d$, leading to an aggregate regret of $\Omega(d^2 \log (1/\epsilon\sqrt{d}))$. We did not do our computations above using this scaled down instance instead of $[0,1]^d$ in order to avoid carrying extra $\sqrt{d}$ terms.
\end{proof}

%!TEX root = main.tex

\section{Computation}\label{sec:computation}

The Projected Volume algorithm described earlier, while yielding optimal regret
with respect to the number of dimensions $d$, can't be implemented as presented in
polynomial time. The reason is that it requires implementing two
steps, both of which involve solving highly nontrivial problems. The first is computing the
centroid, which is known to be {\#P}-hard~\cite{Rademacher07}. The second is finding a
direction along which a convex body $K$ is ``thin'' (i.e. finding a unit
vector $u$ such that $w(K,u)\leq \delta$), for which we are not aware of 
a polynomial time algorithm.

In order to make these problems tractable, we relax the requirements of our
algorithm. More specifically, we will show how our algorithm is robust, in the
sense that using an \emph{approximate centroid}, and finding
\emph{approximately thin directions} does not break the analysis.

In the following subsections, we show how to implement both of these steps.
Then, we put them together into a polynomial time version of our algorithm.

\subsection{Approximating the Centroid}

An approximation of the centroid sufficient for our purposes follows from a
simple application of standard algorithms for sampling points from convex bodies
(hit-and-run~\cite{lovasz1999hit}, ball-walk~\cite{lovasz1999faster}). A similar
application can be found in Bertsimas and Vempala~\cite{BertsimasV04}, where the
authors use approximate centroid computation in order to solve linear programs. 

Our application faces the same issues as in~\cite{BertsimasV04}. Namely, in order to efficiently sample from a convex body, one requires that the body is nearly isotropic. Although the body we start with is isotropic, after cutting or projecting this property is lost. Therefore we require maintaining a linear transformation under which the body ends up being in isotropic position. 
 The many issues encountered when approximating the centroid are carefully handled in~\cite{BertsimasV04}, so we will restate the following result which is implicit there (see Lemma 5 and Theorem 12):
\begin{theorem}[\cite{BertsimasV04}]\label{thm:centroidappx}
Given a $d$-dimensional convex body $K$, one can compute an approximation $z'$ to the centroid $z$ of $K$ in the sense that $\norm{z-z'}\leq \rho$ in $\tilde{O}(d^4/\rho)$ steps of a random walk in $K$.
\end{theorem}

Note that for hit-and-run sampling, one only requires determining the intersection between a line and the given convex body; in our case this only requires one iteration through the inequality constraints determining the body.

\subsection{Finding Approximately Thin Directions}

Instead of exactly recovering directions $u$ satisfying $w(K,u)\leq \delta$, we instead recover \textit{all} the directions along which $w(K,u)\leq \frac{\delta}{\alpha}$, and potentially some along which $\frac{\delta}{\alpha} \leq w(K,u) \leq \delta$. We do this by computing an ellipsoidal approximation of $K$.  Indeed, having access to an ellipsoid $E$ such that
$E \subseteq K \subseteq \alpha E\,{,}$
we can:
\begin{enumerate}
\item  find a direction $u$ such that $w(K,u) \leq \delta$, by checking whether $E$ has a direction $u$ such that $w(E,u)\leq\delta/\alpha$, or
\item decide that $w(K,u) \geq \delta/\alpha$ for all $u$ simply by showing that the smallest directional width of $E$ is greater than or equal to $\delta/\alpha$.
\end{enumerate}
This task can be performed simply by inspecting the eigenvalues of $E$.

A natural notion for such an ellipsoid is the John ellipsoid. However, computing it is NP-hard. Instead, by relaxing the approximation factor, a polynomial time algorithm can be obtained. Such a result is provided in Gr\"otschel et al \cite{GLS93}, which we reproduce below for completeness (see Corollary 4.6.9).

\begin{theorem}[\cite{GLS93}]Given a convex body $K$ containing a ball of radius $r$, and contained inside a ball of radius $R$, along with access to a separation oracle for $K$,
one can compute an ellipsoid $E$ such that
$E \subseteq K \subseteq \sqrt{d}(d+1) E$ using $d^{O(1)} \cdot \log(R/r)$ oracle calls.
\end{theorem}

This immediately yields the following Corollary, which we will use in our algorithmic result.

\begin{corollary}\label{cor:findorcertify}
 Given a convex body $K$ containing a ball of radius $r$, and contained inside a ball of radius $R$, along with a separation oracle for $K$, one can either find a direction $u$ such that $w(K,u) \leq \delta$, or certify that $w(K,u) \geq \delta/(\sqrt{d}(d+1))$ for all $u$ using $d^{O(1)} \cdot \log (R/r)$ oracle calls.
\end{corollary}

\subsection{Obtaining a Polynomial Time Algorithm}
The polynomial time version of our algorithm is very similar to the initial one. The differences that make computation tractable are:
\begin{enumerate}
\item Instead of computing the centroid exactly, we compute the centroid to within distance $\rho = (\epsilon/d)^{O(1)}$, via Theorem~\ref{thm:centroidappx}.
\item Every iteration of the algorithm the set $S_t$ is updated by repeatedly computing the ellipsoidal approximation described in Corollary~\ref{cor:findorcertify}, and adding the direction $u$ corresponding to the smallest eigenvalue of the ellipsoid, if it certifies that $w(K,u)\leq \delta$. When no such direction is found, we know that $w(K,u) \geq \dappx := \delta/{\sqrt{d}(d+1)}$ for all $u$.
\end{enumerate}
A complete description of the new algorithm, along with its analysis, can be found in Appendix~\ref{sec:comp_proofs}. Combining the results in this section, we obtain the following theorem:

\begin{theorem}\label{thm:poly-time}
There exists an algorithm that runs in time $(d/\epsilon)^{O(1)}$ achieving regret $O(d \log (d/\epsilon))$ for the multi-dimensional binary search problem.
\end{theorem}

\bibliographystyle{plain}
\footnotesize
\bibliography{biblio}
\normalsize
\appendix
%!TEX root = main.tex

\section{Deferred Proofs}
\subsection{Proof of Lemma~\ref{lem:dir-grun}}\label{appendix:A1}
  Since width is invariant under rotations and translations we can assume, without loss of generality, that
  $u = -e_1$. Also, since scaling the convex set along the direction of $u$  also scales the corresponding coordinate of the centroid by the same factor, we can assume that the projection of $K$ onto the
  $e_1$ axis is $[0,1]$.
  Using the notation from Theorem \ref{lemma:concave_profile}, we can write the first coordinate of the centroid $z$ as
\[
 z^\top e_1  
= \frac{1}{\vol(K)} \int_K x^\top e_1 \, dx
= \frac{\int_K  x^\top e_1  \, dx}{\int_K 1 \, dx} 
= \frac{\int_0^1 t \cdot r(t)^{d-1}\, dt}{\int_0^1 r(t)^{d-1}\, dt}\,{.}
\]
Our goal is to show that $z^\top e_1 \geq \frac{1}{d+1}$. We will do it in a sequence of
  two steps. To simplify notation, let us define $V := \vol(K)$.

\emph{Step 1: linearize $r$.}
We prove that the linear function $\tilde{r} : [0,1] \rightarrow \mathbb{R}$ given by 
\[
\tilde{r}(t) = (Vd)^{1/(d-1)}\cdot (1-t)
\]
satisfies
  \[
  \int_0^1 t \cdot \tilde{r}(t)^{d-1} dt \leq \int_0^1 t \cdot r(t)^{d-1} dt
  \quad\textnormal{and}\quad
  \int_0^1 \tilde{r}(t)^{d-1} dt = \int_0^1 r(t)^{d-1} dt = V  
  \,{.} 
  \]
We immediately see that the second condition is satisfied, simply by evaluating the integral. Next we show that $\tilde{r}$ satisfies the first condition.

Since by definition, $r$ is supported everywhere over $[0,1]$, it means that $r(1) \geq \tilde{r}(1) = 0$, and therefore $r(0) \leq \tilde{r}(0)$ (since otherwise, by concavity, it would be the case that $r(t)\geq \tilde{r}(t)$ everywhere, and the second identity could not possibly hold). Again, using the concavity of $r$, this implies that there exists a point $p\in[0,1]$ such that $r(t) \leq \tilde{r}(t)$ for all $t\in[0,p]$, and $r(t) \geq \tilde{r}(t)$ for all $t \in [p,1]$.

Hence, we can write 
\[
\int_0^1 t \cdot \left(r(t)^{d-1} - \tilde{r}(t)^{d-1}\right) dt =\int_0^p t \cdot \left(r(t)^{d-1} - \tilde{r}(t)^{d-1}\right) dt + \int_p^1 t \cdot \left(r(t)^{d-1} - \tilde{r}(t)^{d-1}\right) dt \,{,}
\]
where all the coefficients of $t$ from the first term are nonpositive, and all the coefficients of $t$ from the second term are nonnegative. Therefore we can lower bound this integral by
\[
\int_0^p p \cdot \left(r(t)^{d-1} - \tilde{r}(t)^{d-1}\right) dt + \int_p^1 p \cdot \left(r(t)^{d-1} - \tilde{r}(t)^{d-1}\right) dt = p\cdot \left(\int_0^1 r(t)^{d-1}dt - \int_0^1 \tilde{r}(t)^{d-1} dt\right) = 0 \,{,}
\]
which proves that the first condition also holds.

\emph{Step 2: solve for the linear function.}
We can explicitly compute 
\[
\int_0^1 t \cdot \tilde{r}(t)^{d-1} dt = Vd \cdot \int_0^1 t \cdot (1-t)^{d-1} dt = Vd \cdot \frac{1}{d(d+1)} = \frac{V}{d+1}.
\]

Therefore, combining the results from the two steps, we see that
\[
\frac{1}{d+1} = \frac{\int_0^1 t \cdot \tilde{r}(t)^{d-1}dt}{\int_0^1 \tilde{r}(t)^{d-1}dt} \leq \frac{\int_0^1 t \cdot {r}(t)^{d-1}dt}{\int_0^1 {r}(t)^{d-1}dt} = z^\top e_1\,{,}
\]
which yields the desired conclusion.

\subsection{Proof of Lemma~\ref{lemma:approx_grunbaum}}\label{appendix:A2}
  Since our problem is invariant under rotations and translations, let us assume that $u = e_1$, and $z = 0$. Furthermore, notice that our problem is invariant to scaling $K$ along the direction of $u$. Therefore we can assume without loss of generality that
  $[a,1]$ is the projection of $K$ onto the $e_1$ axis. Then, in the notation of Lemma \ref{lemma:concave_profile}, we
  have:
\[
\vol(K_+) = \int_0^1 r(t)^{d-1} dt\,{,} \qquad \vol(K_+^\delta) = \int_{\delta}^1 r(t)^{d-1} dt\,{.}\]

  From Theorem \ref{thm:volume_grunbaum}, we know that $\vol(K_+) \geq \vol(K)/e$.
  We will show that $\vol(K_+^\delta)\geq \vol(K_+) / e $, which yields the sought conclusion. 
  
  From Theorem \ref{thm:directional_grunbaum} we know that $
  w(K,u)/(d+1) \leq 1$. Hence, using our bound on $\delta$, we obtain $\delta \leq 1/(d+1)$. We are left to prove, using the fact that $r$ is a nonnegative concave function, that:
\[
\int_{1/(d+1)}^1 r(t)^{d-1} dt \geq \frac{1}{e} \cdot \int_0^1 r(t)^{d-1} dt\,{.}
\]
  To see that this is true, it is enough to argue that the ratio between the two
  integrals is minimized when $r$ is a linear function $r(t) = c \cdot (t-1)$, for any
  constant $c$; in that case, an explicit computation of the integrals produces
  the desired bound. 
  
To see that the ratio is minimized by a linear function, we proceed in two steps.
First, consider the function $\tilde{r}$ obtained from $r$ by replacing in on the $[1/(d+1), 1]$ interval with a linear function starting at $r(1/(d+1))$ and ending at $0$:
\[
\tilde{r}(t) = 
\begin{cases}
r(t), & \textnormal{if } t\in\left[0,\frac{1}{d+1}\right]\,{,}\\
r\left(\frac{1}{d+1}\right)\cdot\frac{d}{d+1}\cdot(t-1), &\textnormal{if }  t\in\left[\frac{1}{d+1},1\right]\,{.}
\end{cases}
\]
Notice that this function is still concave, and its corresponding ratio of integrals can not be greater than the one for $r$ (since the same value gets subtracted from both integrals when switching from $r$ to $\tilde{r}$).

Next, consider the function
\[
\hat{r}(t) = r\left(\frac{1}{d+1}\right)\cdot\frac{d}{d+1}\cdot(t-1) ,\quad t\in\left[0,1\right]\,{.}
\]
Since $\tilde{r}$ is concave, it is upper bounded by $\hat{r}$ everywhere on $[0,1/(d+1)]$. Therefore, the ratio of integrals corresponding to $\hat{r}$ can only decrease, compared to the one for $\tilde{r}$.

Finally, the result follows from evaluating the integrals for $r(t)=t-1$.
 
\begin{comment}
  From Theorem \ref{thm:volume_grunbaum}, we know $\vol(K_+) \geq \frac{1}{e}
  \vol(K)$. We will show that $\vol(K_+^\delta) \geq \frac{1}{e} \vol(K_+)$.
  From Theorem \ref{thm:directional_grunbaum} we know that $b \geq
  \frac{w(K,u)}{d+1}$ so $\delta \leq \frac{b}{d+1}$. We also know by Lemma
  \ref{lemma:concave_profile}  that $g(t) = r(t)^{d-1}$ for a concave function
  $g$. So we are left to prove that if $r$ is concave, then:
  $$\int_{b/(d+1)}^b r(t)^{d-1} dt \geq \frac{1}{e} \cdot \int_0^b r(t)^{d-1} dt\,{.}$$
  To see that this is true, it is enough to argue that the ratio between the two
  integrals is minimized when $r$ is a linear function $r(t) = c (r-b)$ for any
  constant $c$ in which case a explicit computation of the integrals produces
  the desired bound. To see that the ratio is minimized by a linear function,
  for any $r$ replace the by $r_1(t) = r(t) $ for $t \in [0, \frac{b}{d+1}]$ and
  $r_1(t) = r\left(\frac{b}{d+1}\right) \frac{d+1}{db} (t-b)$ and observe that
  the function is still concave and ratio of integrals can't increase since the
  same value is subtracted from both integrals. Then
  replace by $r_2(t) =  r\left(\frac{b}{d+1}\right) \frac{d+1}{db} (t-b)$ for $t
  \in [0,b]$ and notice that the ratio can't decrease since one of the integrals
  strictly increases and the other remains constant. Since $r_2(t)$ is linear,
  we can compute the ratio explicitly.
\end{comment}

\subsection{Proof of Lemma~\ref{lemma:projected_grunbaum}}\label{appendix:A3}
Since the problem is invariant under rotations and translations, we can assume without loss of generality that $z = 0$, $S = \{e_1, \hdots, e_k\}$
and $L = \span\{e_{k+1}, \hdots, e_{n}\}$. For every vector
$x$ we will consider the projections of $x$ onto the two corresponding subspaces, $x_S = (x_1, \hdots, x_k)$ and $x_L = (x_{k+1}, \hdots,
x_n)$. For simplicity, will also use the notation $K_L := \Pi_L K$.

The proof consists of four steps.

\emph{Step 1: the direction $u$ has a large component in $L$.} 
Since $w(\Cyl(K,S),u)
\geq \epsilon$, and $z = 0$ is the centroid of the cylinder,
there must exist $y \in \Cyl(K,S)$ such that $\abs{u^\top y} =
\abs{u^\top (y-z)} \geq \frac{\epsilon}{2}$. Therefore $\abs{u_S^\top y_S} + \abs{u_L^\top y_L} \geq \frac{\epsilon}{2}$. Since the width of
$\Cyl(K,S)$ is at most $\delta$ along all small directions, we have $\norm{y_S}_\infty \leq \delta$.
Therefore, by Cauchy-Schwarz, 
\[
\norm{u_L} \norm{y_L} \geq \abs{u_L^\top y_L} \geq \frac{\epsilon}{2} - k\delta\,{.}
\]
Now, remember that since $y \in \Cyl(K,S)$, $K$ is contained inside the unit ball, and all the small directions have length at most $\delta$, it must be that $\norm{y} \leq 1 + k\delta$. Since this implies the same upper bound on $\norm{y_L}$, combining with the bound above we see that
\[
\norm{u_L} \geq \frac{\epsilon/2-k\delta }{1+k\delta} \geq \frac{\epsilon/2 - \epsilon^2 / (16(d+1))}{1 + \epsilon^2 / (16(d+1))}\geq \frac{\epsilon}{4}\,{.}
\]

\emph{Step 2: lower bound the width of $K_L$ along the direction of ${u}_L$}. Let $\hat{u}_L = u_L
/ \norm{u_L}$ be the unit vector in the direction $u_L$. 
We know by the last step that 
\[
w(K_L, u_L) \geq \abs{\hat{u}_L^\top
y_L} \geq \abs{u_L^\top y_L} \geq \frac{\epsilon}{2} - k\delta \geq \frac{\epsilon}{4}\,{.}
\]

\emph{Step 3: show that for all $x \in K_+$, one has $\hat{u}_L^\top x_L \geq
-{\epsilon}/{(4(d+1)^2)}$}. If $x \in K_+$, then $u_L^\top x_L + u_S^\top x_S
\geq 0$. Since $\norm{x_S}_\infty \leq \delta$, we have $u_L^\top x_L \geq
 - k \delta $. Hence 
 \[
\hat{u}_L^\top x_L \geq -\frac{k \delta}{\norm{u_L}} \geq
 -\frac{4 d \delta}{\epsilon} \geq -\frac{\epsilon}{4(d+1)^2}\,{,}
 \]
  where we used the fact that $\delta \leq
{\epsilon^2}/{(16 d(d+1)^2)}$.

\emph{Step 4: upper bound the volume of $\Pi_L K_+$.} From the previous step, we know
that if
 $x \in \Pi_L (K_+)$,
 then 
 $x_L \in \{x_L \in K_L\vert \hat{u}_L^\top x_L \geq {-\epsilon}/{(4(d+1))}\}$. Therefore:
\[
\vol(\Pi_L K_+) \leq \vol(K_L) - \vol \left( \left\{x_L \in K_L\bigg\vert
(-\hat{u}_L)^\top x_L \geq \frac{\epsilon}{4 (d+1)^2} \right\}\right) \leq \vol(K_L)
\cdot \left( 1-\frac{1}{e^2}\right)\,{,}
\]
where the first inequality follows from the previous step, since 
$$\Pi_L K_+
\subseteq \left\{x_L \in K_L\bigg\vert (-\hat{u}_L)^\top x_L \leq \frac{\epsilon}{4
(d+1)^2} \right\}\,{.}$$ The second inequality follows from Lemma
\ref{lemma:approx_grunbaum}, since in Step $2$  we showed that in this proof, we meet the
conditions of that lemma. We note that it is very important that $z$ is the centroid of
$\Cyl(K,S)$ and not the centroid of $K$, since the application of Lemma
\ref{lemma:approx_grunbaum} relies on the fact the projection of $z$ onto the
subspace $L$ is the centroid of $K_L$.

%!TEX root = main.tex
\subsection{Proof of Lemma \ref{lemma:centroid_simplex}}

Let $z_i = \frac{1}{\vol(\Delta(s))} \int_{\Delta(s)} x_i dx$ be the $i$-th
component of the centroid of $\Delta(s)$. So if $s_{-i}$ is the vector in
$\R^{k-1}$ obtained from $s$ by removing the $i$-th component, then the
intersection of $\Delta(s)$ with the hyperplane $x_i = a$ can be written as: $\{x\vert~
x_i = a,~ x_{-i} \in \Delta(\frac{s_{-i}}{1-{a}/{s_i}}) \}$. Therefore, we
can write the integral defining $z_i$ as:
$$z_i =  \frac{1}{\vol(\Delta(s))} \int_0^{s_i}
x_i \vol\left(\Delta\left(\frac{s_{-i}}{1-\frac{x_i}{s_i}}\right)\right) dx_i =
 \frac{1}{\vol(\Delta(s))} \int_0^{s_i} x_i
\vol\left(\Delta\left({s_{-i}}\right)\right)\cdot \left( 1-\frac{x_i}{s_i}
\right)^{k-1} dx_i
$$
since scaling each coordinate a constant factor scales the volume by this
constant powered to the number of dimensions. Solving this integral, we get:
$$z_i = \frac{ \vol(\Delta(s_{-i})) }{ \vol(\Delta(s)) } \cdot
\frac{s_i^2}{k(k+1)}~.$$
We can apply the same trick to compute the volume:
$$\vol(\Delta(S)) = \int_0^{s_i} \vol\left(\Delta\left({s_{-i}}\right)\right)\cdot
\left( 1-\frac{x_i}{s_i} \right)^{k-1} dx_i = \vol(\Delta(s_{-i})) \cdot \frac{s_i}{k}~.$$
Substituting the volume $\vol(\Delta(S))$ in $z_i$ we get $z_i =
\frac{s_i}{k+1}$.

\subsection{Proof of Lemma
\ref{lemma:step_counterexample}}\label{appendix:proof_step_counterexample}

We break the sequence of directions chosen by nature in three parts. We will
show that the first part alone has regret $O(k \log (\frac{1}{\epsilon}))$ and
the other two parts will be used to bring the knowledge set to the desired
format. We won't specify the exact value of $\theta$. We only assume that
$\theta$ is an arbitrary point in the final knowledge set produced.

\emph{Step 1}: Nature picks $\Omega(k \log(\frac{1}{\epsilon}))$ vectors in the direction $e_k$,
choosing the $K_+$ side.

The knowledge set is initially $\Delta(s) \times [0,1]^{d-k}$ with centroid at
$\left(\frac{s}{k+1}, \frac{1}{2}, \hdots, \frac{1}{2} \right)$. The set
obtained by cutting through this point using a hyperplane orthogonal to $e_k$
can be described as
$$\left\{ x \in \mathbb R^d: \quad x_{k} \geq \frac{s_k}{k+1}, \qquad \sum_{i=1}^k \frac{x_i}{s_i} \leq 1, \qquad
0 \leq x_i \leq 1\right\},$$
which is, up to translation, equal to the set $\Delta((1-\frac{1}{k+1})s) \times
[0,1]^{d-k}$. By applying such cuts $\Omega(k \log(\frac{1}{\epsilon}))$ we are left with a
set $\Delta(\hat{s}) \times [0,1]^{d-k}$ where $0 \leq \hat{s} \leq \epsilon$.
Since we assumed that $\theta$ is in the last knowledge set while $s_k \geq
2\epsilon \frac{k}{k+1}$ we must be incurring one unit of regret, so we must
have incurred at least $\Omega(k \log(\frac{1}{\epsilon}))$ regret.

\emph{Step 2}: Nature picks a single vector in the direction 
$v = \left(
\frac{k+1}{2k}\cdot\frac{1}{\hat s_1}, \hdots,
\frac{k+1}{2k}\cdot\frac{1}{\hat s_k},
1, 0, \hdots, 0 \right)$, choosing the $K_-$ side. 

Since the centroid is
$z = (\frac{\hat s}{k+1}, \frac{1}{2}, \hdots, \frac{1}{2})$ the half-space defining
$K_-$ is given by: $v^\top x \leq v^\top z = 1$, therefore $K_-$ is described
by:
$$K_- = \left\{ x \in \mathbb R^d: \quad  \sum_{i=1}^k \frac{x_i}{\hat s_i}\frac{k+1}{2k} + x_{k+1} \leq 1,
\qquad \sum_{i=1}^k \frac{x_i}{\hat s_i} \leq 1, \qquad
0 \leq x_i \leq 1\right\}$$
To understand the shape of $K_-$ it is useful to decompose it in two parts based
on the value of $x_{k+1}$. Let $y = 1-\frac{1}{2}\frac{k+1}{k}$ which is a
quantity between $0$ and $\frac{1}{2}$.
\begin{itemize}
\item for $x \in K_-$ with $x_{k+1} \geq y$ the constraint
$ \sum_{i=1}^k \frac{x_i}{\hat s_i} \leq 1$ is implied by
$\sum_{i=1}^k \frac{x_i}{\hat s_i}\frac{k+1}{2k} + x_{k+1} \leq 1 $, since we can re-write the
second constraint as: $\sum_{i=1}^k \frac{x_i}{\hat s_i}(1-y) \leq 1-x_{k+1}
\leq 1-y$. This means in particular that $\{x \in K_-:~ x_{k+1} \geq y\}$ is
equal, up to translation to $\Delta(\hat s, 1-y) \times [0,1]^{d-k-1}$.
\item For $x \in K_-$ with $x_{k+1} \leq y$ then the constraint $\sum_{i=1}^k
\frac{x_i}{\hat s_i}\frac{k+1}{2k} + x_{k+1} \leq 1 $ is implied by
$ \sum_{i=1}^k \frac{x_i}{\hat s_i} \leq 1$ since $\sum_{i=1}^k \frac{x_i}{\hat
s_i}(1-y) \leq 1-y \leq 1-x_{k+1}$. In particular, this means that $\{x \in K_-:~
x_{k+1} \leq y\}$ is the set $\Delta(\hat s) \times [0,y] \times [0,1]^{d-k-1}$.
\end{itemize}

\emph{Step 3}: Nature picks $r$ vectors in direction $e_{k+1}$ choosing  the $K_+$ side, where $r$ 
will be decided later.

After Step 2, the set is a concatentation of $\Delta(\hat s) \times [0,y] \times
[0,1]^{d-k-1}$ and $\Delta(\hat s, 1-y) \times [0,1]^{d-k-1}$ as displayed in
Figure \ref{fig:lower_1}. 
By cutting in the $e_{k+1}$ direction, we will eventually be left
only with the $\Delta(\hat s, 1-y) \times [0,1]^{d-k-1}$ part of the set.
Pick $r$ to be the minimum value such that this happens.  Since the volume of
the sections along the $x_{k+1}$ dimension are non-increasing, the set
after the cut must keep at least half of the width along $e_{k+1}$. Therefore,
after $r$ cuts, we must be left with $\Delta(s') \times [0,1]^{d-k-1}$ where $s'
\in \R^{k+1}$ and $\frac{1}{4} \leq \frac{1-y}{2} \leq s'_{k+1} \leq 1-y$.

%!TEX root = main.tex
\section{Polynomial Time Algorithm }\label{sec:comp_proofs}
\paragraph{Correctness.} Correctness of this algorithm follows from a simple modification of our original analysis. In order to tolerate the fact that the centroid produced by the sampling scheme is only approximate, we need to resort to the Approximate Gr\"{u}nbaum Theorem  (see Lemma~\ref{lemma:approx_grunbaum}) in order to track the decrease in volume, and also to an approximate version of the Directional Gr\"{u}nbaum Theorem (see Lemma~\ref{thm:appxdirectional} below), in order to argue that directional widths still do not decrease faster than they are supposed to.

\begin{lemma}[Approximate Directional Gr\"{u}nbaum]\label{thm:appxdirectional}
 Let $K$ be a convex body with centroid $z$. Let $z'$ be an approximate centroid in the sense that $\norm{z-z'}\leq \rho$. 
Then for every vector $u\neq 0$, the set $K_+ = \{x\vert u^{\top} (x-z') \geq 0 \}$ satisfies
  $$\frac{1}{d+1} \cdot w(K,v) - \rho \cdot \max\left( 1, \frac{w(K,v)}{w(K,u)}
  \right) \leq w(K_+,v) \leq w(K,v)$$ for any unit vector $v$.
\end{lemma}
\begin{proof}[Proof sketch]
The analysis follows from minor modifications in the analysis of Theorem~\ref{thm:directional_grunbaum}. First we modify Lemma~\ref{lem:dir-grun} in order to show that $$\frac{1}{d+1} w(K,u) -\rho \leq w(K_+,u) \leq w(K,u)\,{.}$$
Indeed, since $\norm{z-z'} \leq \rho$, taking a cut perpendicular to $u$ that passes through $z'$ instead of $z$ changes the directional width along $u$ by at most $\rho$. Therefore the bound above holds in the worst case.
Second, we consider the three cases considered in the proof. In the first two cases, we have $w(K_+,v) \geq \frac{1}{d+1} \cdot w(K,v)-\rho$ via the previous bound. In the third case, let $\lambda^+$ and $\lambda^-$ defined similarly.  Then we have
$\frac{1}{d+1} w(K,u) - \rho \leq \lambda^+ \cdot w(K,u)$ and
similarly for $\lambda^-$. Therefore $\min\{\lambda^+, \lambda^-\} \geq \frac{1}{d+1} - \frac{\rho}{w(K,u)}$. Finally, this yields
$$w(K_+,v)\geq w(K,v) \cdot \left(\frac{1}{d+1} - \frac{\rho}{w(K,u)}\right)  \,{,}$$ and our conclusion follows.
\end{proof}

Similarly, we require a robust version of projected Gr\"{u}nbaum, which we sketch below.

\begin{lemma}[Approximate Projected Gr\"{u}nbaum]\label{appx-proj-gr}
  Let $K$ be a convex set contained in the ball of radius $1$, and let
  $S$ be a set of orthonormal
  vectors along which $w(K,s) \leq \delta \leq \frac{\epsilon^2}{32 d(d+1)^2}$,
  for all $s \in S$. Let $L$ be the
  subspace orthogonal to $S$, and let $\Pi_L$ be the projection operator onto that
  subspace.
  If $u$ is a direction along which $w(\Cyl(K,S),u) \geq \epsilon$, $z$ is the
  centroid of the cylindrified body $\Cyl(K,S)$, $z'$ satisfies $\norm{z-z'}\leq
  \rho := \frac{\epsilon}{8(d+1)^2}$, and
  $K_+ = \{x \in K:~ u^\top (x-z') \geq 0\}$, then:
  $$\vol(\Pi_L K_+) \leq \left( 1-  \frac{1}{e^2} \right) \cdot \vol(\Pi_L K)\,{,}$$
  where $\vol(\cdot)$ corresponds to the $(n-\abs{S})$-dimensional volume on the
  subspace $L$.
\end{lemma}
\begin{proof}[Proof sketch]
The proof follows the same steps as the proof of Lemma~\ref{lemma:projected_grunbaum}. Below we sketch the essential differences, and show how they affect the analysis.

The first two steps are identical, since they do not involve the perturbed centroid $z'$. For the third step, we proceed identically to show that
\[
\hat{u}_L^\top x_L \geq -\frac{4d\delta}{\epsilon} \geq
  -\frac{\epsilon}{8(d+1)^2}\,{,}
\]
where we used $\delta \leq \frac{\epsilon^2}{32d(d+1)^2}$. 

Finally, for the fourth step we use the fact that if $x \in \Pi_L(K_+)$ then
$x_L \in \{x_L \in K_L \vert \hat{u}_L^\top x_L \geq - \epsilon/(8(d+1)^2) - \rho \}$. Hence
\[
\Pi_LK_+ \subseteq \left\{ x_L \in K_L \bigg\vert (-\hat{u}_L)^\top x_L \leq
  \frac{\epsilon}{4(d+1)^2} \right\}\,{.}
\]
and thus we obtain the same bound on $\vol(\Pi_L K_+)$.
\end{proof}

Putting everything together, we can show that the algorithm using these approximation primitives yields the same regret asymptotically.
The constants we will use throughout the our algorithm will be 
$\dappx = \delta / (\sqrt{d}(d+1)) = \epsilon^2 / (16 d^{1.5} (d+1)^3)$, and $\rho = \dappx^2 / (2(d+1))$.
The two key results required for our robust analysis are:
\begin{enumerate}
\item
If the cardinality of $S$ does not increase, then
\[
\vol(\Pi_{L_t} K_{t+1}) \leq \left( 1 - 1/e^2 \right) \vol(\Pi_{L_t} K_t)\,{.}
\]
This is given by the approximate projected Gr\"unbaum theorem (Lemma~\ref{appx-proj-gr}).
\item
When adding an extra direction to $S$,  we know that $w(K_t, u) \geq \dappx$, for all $u \in L_t$.

Then by Lemma \ref{thm:appxdirectional} after we cut $K_t$ we have that for any vector $v \in L_t$, 
\[
w(K_{t+1},v)
    \geq \frac{w(K_t,v)}{d+1} - \rho \cdot \max\left(1,
    \frac{w(K_t,u)}{w(K_t,v)} \right)
\geq \frac{\dappx}{d+1} - \rho \cdot \frac{1}{\dappx} 
\geq \frac{\dappx}{2(d+1)}\,{,}
\]
by our choice of $\rho = \dappx^2 / (2(d+1))$.
 So applying the Cylindrification Lemma (Lemma
\ref{lemma:cylindrification}) we obtain that the volume of the convex body projected onto the new subspace of large directions $L'$ is bounded by
\[
\vol(\Pi_{L'} K_{t+1}) \leq \frac{d(d+1)}{\dappx  / (2(d+1))} \vol(\Pi_{L_t} K_{t+1})= \frac{32 d^{1.5}(d+1)^3}{\delta} \vol(\Pi_{L_t} K_{t+1})\,{.}
\]
This follows just like before from Lemma~\ref{lemma:cylindrification}. Our method of finding thin directions based on the approximate John ellipsoid (Corollary~\ref{cor:findorcertify}) guarantees that all directional widths in the large subspace $L$ are at least $\dappx$. 
Therefore the blow up in volume is at most by a factor of ${(32 d^{1.5}(d+1)^3)}/{\delta}$.
\end{enumerate}

Since all the new bounds are within polynomial factors from the ones used in the analysis using exact centroids, by plugging in the old analysis, we easily obtain the same regret, up to constant factors.

\paragraph{Running time.} For the running time analysis, note that the centroid approximation can be implemented using $\tilde{O}(d^4 / \rho ) = (d/\epsilon)^{O(1)}$ calls to the separation oracle for the convex body. Such a separation oracle needs to take into account both the linear inequalities added during each iteration, and the at most $d$ projections. Such an oracle can be implemented by maximizing a linear functional over a set determined by the intersection between the initial unit ball and the linear constraints (whose number is bounded by the number of iterations of the algorithm $\tilde{O}(d \log(1/\epsilon))$; therefore this step can be implemented in polynomial time, and therefore all the centroid approximation steps require time $(d/\epsilon)^{O(1)}$.

The routine for finding the thin directions will be called at least once every iteration, and will find a thin direction at most $d$ times. Therefore this contributes $d^{O(1)} \log(R/r) \cdot \log(1/\epsilon)$ to the running time, where $r$ is a lower bound on the smallest ball contained in the body, while $R$ is an upper bound. From the setup we have $R=1$; also, since we are finished after $\tilde{O}(d\log(1/\epsilon))$ iterations, and each iteration shrinks the smallest directional width by at most a factor of $d^{O(1)}$, according to Lemma~\ref{lemma:large_ball}, we have that at all times the body will contain a ball of radius $d^{-\Omega(d)}$. Therefore the running time contribution of the routine required for finding thin directions is $d^{O(1)}\log(1/\epsilon)$.

All the other steps require at most polynomial overhead, therefore the total running time is $(d/\epsilon)^{O(1)}$.

%!TEX root = main.tex

\section{Relation to standard online learning problems}\label{sec:comparison}

We now discuss in detail the relationship between multidimensional binary search (MBS) and standard problems in online learning. We start by describing the problem of
learning a halfspace with a margin (LHM).

\subsection{Learning Halfspaces with Margin}
As in our multidimensional binary search problem, this problem starts by assuming there exists a fixed but unknown vector of weights $\theta \in
\R^d$ with $\norm{\theta} = 1$. We receive vectors of
features $u_t \in \R^d$ with $\norm{u_t} = 1$ in an online fashion.
The algorithm needs to guess on which side of
the hyperplane defined by $\theta$ the feature vector $x_t$ lies. 
If $\hat{y}_t \in \{-1, +1\}$ is the guess, the loss/regret is $1$ 
if $\hat{y}_t \cdot u_t^\top \theta \leq 0$ and
zero otherwise. The feedback after each guess is $\sgn{u_t^\top \theta}$, which also tells us how
much regret we incur in each step. The problem also comes with a promise. All feature vectors satisfy a condition related to the hyperplane being learned: 
$\abs{\theta^\top u_t} \geq \gamma$.

\subsection{Halving Algorithm}

We start by discussing the Halving Algorithm for LHM. While not the most popular method
for this problem, it is the one closest to our approach. In this algorithm, we also maintain a
knowledge set of all candidate vectors $\theta$. We start with $K_0 = \{ \theta|~
\norm{\theta} \leq 2 \}$, which is slightly larger than usual to guarantee we have a
small ball around the real $\theta$. For each incoming $u_t$, we find
$x_t$ such that 
$$\vol(K_t \cap \{\theta|~ u_t^\top \theta \leq x_t\}) = \vol(K_t
\cap \{\theta|~ u_t^\top \theta \geq x_t\})$$
and guess $\hat y_t = \sgn{x_t}$.
If the guess is correct, we don't incur any loss and we don't update the
knowledge set. If the guess is wrong, we update the knowledge set either to
$K_{t+1} = K_t \cap \{\theta|~ u_t^\top \theta \leq 0\}$ (if we guessed $y_t =
+1$) or $K_{t+1} = K_t \cap \{\theta|~ u_t^\top \theta \geq 0\}$ (if we guessed
$y_t = -1$). After each incorrect guess, we have that $\vol(K_{t+1})
\leq \frac{1}{2} \vol(K_t)$. 

If we can provide a lower bound to the final volume, we are able to provide a
mistake bound. By the promise we know that for the real $\theta$,
$\abs{\theta^\top u_t} \geq \gamma$. If we add a constraint $\theta^\top u_t
\geq 0$ to the knowledge set, then this constraint should be satisfied for any
$\theta'$ such that $\norm{\theta - \theta'} \leq \gamma$, since:
$$u_t^\top \theta' = u_t^\top \theta + u_t^\top(\theta-\theta') \leq \gamma -
\norm{\theta-\theta'} \cdot \norm{u_t} \geq 0.$$
Therefore the algorithm never removes a ball of radius $\gamma$ around the real
$\theta$ from the knowledge set. Therefore, the volume of the knowledge set is
always at least $\Omega(\gamma^d)$, which yields a $O(d \log(1/\gamma))$ mistake
bound.

\emph{Can we adapt halving to MBS?} The key obstacle is that we do not know if we
made a mistake or not in MBS. If we knew whether we made a mistake or not, we
could keep a knowledge set as before. For each $u_t$, we again would choose $x_t$ such
that:
$$\vol(K_t \cap \{\theta|~ u_t^\top \theta \leq x_t\}) = \vol(K_t
\cap \{\theta|~ u_t^\top \theta \geq x_t\}).$$
If we do not make a mistake, we again don't update the knowledge set.
As in the previous problem, if we make a
mistake, we update the knowledge set to one of half the volume. If we guessed
$u_t^\top \theta = x_t$ and it was a mistake, we can update either to $K_t \cap
\{\theta|~ u_t^\top \theta \leq 0\}$  or $K_t \cap \{\theta|~ u_t^\top \theta \geq
0\}$ depending on the feedback. If we update to $K_t \cap\{\theta|~ u_t^\top
\theta \geq 0\}$, we know by the fact that we made a mistake that $u_t^\top\theta
\geq \epsilon$. So any $\theta'$ within a ball of radius $\epsilon$ of the
real $\theta$ will never be removed from the knowledge set by the same argument
as before, using $\epsilon$ instead of $\gamma$. This gives us by the same
volume argument a $O(d \log(1/\epsilon))$ regret algorithm.

However, this algorithm is \emph{not} implementable in our setting because we
do not know if we made a mistake or not. We could always cut the knowledge set after
the feedback, but then we are no longer guaranteed to preserve a ball of radius
$\epsilon$ around the true $\theta$ and no longer can make the volume argument.

Can we always cut and use a different analysis to show the same regret bound?
In Section \ref{sec:whycyl}, we argued that if we always cut through the centroid
we might incur $\Omega(d^2 \log(1/\epsilon))$ regret. The same argument holds if
we cut in half the volume instead of cutting through the centroid.
Therefore, the answer is no if we want to obtain a linear bound on $d$. 

\subsection{Online Convex Optimization}

The other family of algorithms for learning halfspaces with a margin is based on
online convex optimization. Using the classic Perceptron algorithm, one can
obtain a mistake bound of
$O(1/\gamma^2)$ for learning halfspaces with a margin  and using the closely
related Winnow algorithm we
obtain a much better dependency on $\gamma$ at the expense of a dependency in
the dimension: $O(d \log(d/\gamma)).$\footnote{We also point out that Perceptron-style algorithms like the one of~\cite{DunaganV08}, despite exhibiting the apparently useful number of $O(d \log d/\epsilon)$ iterations, are
 not applicable to this problem. Crucially, the problem studied in those settings is offline: it requires knowing all the vectors given by the adversary ahead of time. } We refer the reader to Section 3.3 in~\cite{shalev2012online} for a comprehensive exposition. 

Instead of keeping a candidate set of possible
values of $\theta$, those algorithms keep a single value $\hat \theta \in \R^d$ and for each
incoming vectors of features $u_t$, we make a prediction according to $u_t^\top
\hat \theta$. If we make a mistake, we perform a first-order update to
$\hat\theta$ according to a surrogate loss function. For the learning
halfspaces problem, the actual loss function is replaced by a convex function
that is an upper bound to the actual loss function and for which we can compute the
gradient as a function of $u_t$ and $y_t$. That function is taken to be the
hinge loss.  Depending on which type of first order update we do, we get a
different algorithm: the Perceptron is obtained by performing a gradient
descent step and Winnow is obtained by performing an exponentiated gradient
step.

Both algorithms depend on not performing an update whenever we
do not incur a mistake. 
For our problem, since we do not know if we made a mistake or not, to instantiate
the online convex optimization framework we would need to perform an update in
each iteration and there does not seem to be an adequate surrogate
for our loss function $\ell(\hat \theta, u_t) = {\mathbf
1} \{ {\lvert{u_t^\top\hat\theta- u_t^\top\theta }\rvert> \epsilon } \}$.

\end{document}